\documentclass[11pt]{article}%
\DeclareMathAlphabet{\mathbbold}{U}{bbold}{m}{n}
\usepackage{amsmath,amsfonts,amssymb}
\usepackage{amsthm}
\usepackage{graphicx}
\usepackage[margin=1in,a4paper]{geometry} 
\usepackage[pagebackref]{hyperref}
\hypersetup{
    colorlinks=true, 
    linkcolor=blue, 
    citecolor=blue, 
    urlcolor=blue 
}

\usepackage[linesnumbered,ruled,vlined]{algorithm2e}
\SetKwInput{KwInput}{Input}                
\SetKwInput{KwOutput}{Output}
\SetKwFor{RepTimes}{repeat}{times}{end}

\usepackage{microtype}
\usepackage{thm-restate}
\usepackage[capitalize,nameinlink]{cleveref} 
\usepackage{mathtools}

\usepackage{tikz}
\usepackage{pgfplots}
\pgfplotsset{compat=1.13}
\usepackage{enumerate}
\usepackage{qcircuit}

\renewcommand{\backref}[1]{}

\renewcommand{\backrefalt}[4]{%
\ifcase #1 %
\or
[p.\ #2]%
\else
[pp.\ #2]%
\fi}

\providecommand{\U}[1]{\protect\rule{.1in}{.1in}}
\newtheorem{theorem}{Theorem}

\newtheorem{conjecture}[theorem]{Conjecture}

\newtheorem{definition}[theorem]{Definition}

\newtheorem{fact}[theorem]{Fact}
\newtheorem{lemma}[theorem]{Lemma}

\newtheorem{proposition}[theorem]{Proposition}

\newcommand{\identity}{I}
\newcommand{\E}{\mathop{\mathbb{E}}}

\newcommand{\bra}[1]{\langle #1 |}
\newcommand{\ket}[1]{| #1 \rangle}
\newcommand{\braket}[1]{\langle #1 \rangle}

\newcommand{\eps}{\varepsilon}
\newcommand{\Reals}{\mathbb{R}}
\newcommand{\Naturals}{\mathbb{N}}

\newcommand{\mathify}[1]{\ifmmode{#1}\else\mbox{$#1$}\fi}
\newcommand{\abs}[1]{\mathify{\left| #1 \right|}}

\newcommand{\Unitaries}{\mathbb{U}}
\newcommand{\States}{\mathbb{S}}

\newcommand{\Oracle}{\mathcal{O}}
\newcommand{\UOracle}{\mathcal{U}}
\newcommand{\COracle}{\mathcal{C}}
\newcommand{\poly}{\mathrm{poly}}
\newcommand{\adv}{\mathrm{adv}}
\newcommand{\negl}{\mathrm{negl}}

\newcommand{\diamondnorm}[1]{\left|\left| #1 \right|\right|_\diamond}
\newcommand{\tr}{\mathrm{Tr}}
\newcommand{\TD}{\mathrm{TD}}
\newcommand{\Alg}{\mathcal{A}}
\newcommand{\AlgB}{\mathcal{B}}
\newcommand{\Ver}{\mathcal{V}}
\newcommand{\VerW}{\mathcal{W}}
\newcommand{\Lang}{\mathcal{L}}
\newcommand{\Key}{\{0,1\}^\kappa}
\renewcommand{\Re}{\mathrm{Re}}
\DeclareMathOperator*{\argmax}{arg\,max}
\newcommand{\wt}{\operatorname{wt}} 
\newcommand{\Dom}{\mathrm{Dom}}

\begin{document}

\title{Quantum Pseudorandomness and Classical Complexity}

\author{William Kretschmer\thanks{University of Texas at Austin. \ Email:
\texttt{kretsch@cs.utexas.edu}. Supported by an NDSEG Fellowship.}}

\date{}
\maketitle

\begin{abstract}
We construct a quantum oracle relative to which $\mathsf{BQP} = \mathsf{QMA}$ but cryptographic pseudorandom quantum states and pseudorandom unitary transformations exist, a counterintuitive result in light of the fact that pseudorandom states can be ``broken'' by quantum Merlin-Arthur adversaries. We explain how this nuance arises as the result of a distinction between algorithms that operate on quantum and classical inputs. On the other hand, we show that \textit{some} computational complexity assumption is needed to construct pseudorandom states, by proving that pseudorandom states do not exist if $\mathsf{BQP} = \mathsf{PP}$. We discuss implications of these results for cryptography, complexity theory, and shadow tomography.
\end{abstract}

\section{Introduction}
Pseudorandomness is a key concept in complexity theory and cryptography, capturing the notion of objects that appear random to computationally-bounded adversaries. Recent works have extended the theory of computational pseudorandomness to quantum objects, with a particular focus on quantum states and unitary transformations that resemble the Haar measure \cite{JLS18-prs,BS19-binary,BFV20-prs-wormhole}.

Ji, Liu, and Song \cite{JLS18-prs} define a \textit{pseudorandom state} (PRS) ensemble as a keyed family of quantum states $\{\ket{\varphi_k}\}_{k \in \Key}$ such that states from the ensemble can be generated in time polynomial in $\kappa$, and such that no polynomial-time quantum adversary can distinguish polynomially many copies of a random $\ket{\varphi_k}$ from polynomially many copies of a Haar-random state. They also define an ensemble of \textit{pseudorandom unitary transformations} (PRUs) analogously as a set of efficiently implementable unitary transformations that are computationally indistinguishable from the Haar measure. These definitions can be viewed as quantum analogues of pseudorandom generators (PRGs) and pseudorandom functions (PRFs), respectively. The authors then present a construction of PRSs assuming the existence of quantum-secure one-way functions, and also give a candidate construction of PRUs that they conjecture is secure.

Several applications of pseudorandom states and unitaries are known. PRSs and PRUs are useful in quantum algorithms: in computational applications that require approximations to the Haar measure, PRSs and PRUs can be much more efficient than $t$-designs, which are information-theoretic approximations to the Haar measure that are analogous to $t$-wise independent functions.\footnote{$t$-designs are also sometimes called ``pseudorandom'' in the literature, e.g.\ \cite{WBV08-pseudorandom,BHH16-designs-short}. We emphasize that $t$-designs and PRSs/PRUs are fundamentally different notions and that they are generally incomparable: a $t$-design need not be a PRS/PRU ensemble, or vice-versa.} Additionally, a variety of cryptographic primitives can be instantiated using PRSs and PRUs, including quantum money schemes, quantum commitments, secure multiparty communication, one-time digital signatures, some forms of symmetric-key encryption, and more \cite{JLS18-prs,AQY22-prs,MY22-prs,BCQ23-efi,MY22-owq,HMY23-qpke}. Finally, Bouland, Fefferman, and Vazirani \cite{BFV20-prs-wormhole} have established a fundamental connection between PRSs and any possible resolution to the so-called ``wormhole growth paradox'' in the AdS/CFT correspondence.

\subsection{Main results}
Given the importance of pseudorandom states and unitaries across quantum complexity theory, cryptography, and physics, in this work we seek to better understand the theoretical basis for the existence of these primitives. We start with a very basic question: what hardness assumptions are necessary for the existence of PRSs,\footnote{Note that PRUs imply PRSs, so we focus only on PRSs for this part.} and which unlikely complexity collapses (such as $\mathsf{P} = \mathsf{PSPACE}$ or $\mathsf{BQP} = \mathsf{QMA}$) would invalidate the security of PRSs? Viewed another way, we ask: what computational power suffices to distinguish PRSs from Haar-random states?

At first glance, it appears that an ``obvious'' upper bound on the power needed to break PRSs is $\mathsf{QMA}$, the quantum analogue of $\mathsf{NP}$ consisting of problems decidable by a polynomial-time quantum Merlin-Arthur protocol (or even $\mathsf{QCMA}$, where the witness is restricted to be classical). If Arthur holds many copies of a pure quantum state $\ket{\psi}$ that can be prepared by some polynomial-size quantum circuit $C$, then Merlin can send Arthur a classical description of $C$, and Arthur can verify via the swap test that the output of $C$ approximates $\ket{\psi}$. By contrast, most Haar-random states cannot even be approximated by small quantum circuits. So, in some sense, PRSs can be ``distinguished'' from Haar-random by quantum Merlin-Arthur adversaries.

There is a subtle problem here, though: $\mathsf{QMA}$ is defined as a set of decision problems where the inputs are \textit{classical} bit strings, whereas an adversary against a PRS ensemble inherently operates on a \textit{quantum} input. As a result, it is unclear whether the hardness of breaking PRSs can be related to the hardness of $\mathsf{QMA}$, or any other standard complexity class. Even if we had a proof that $\mathsf{BQP} = \mathsf{QMA}$, this might not give rise to an efficient algorithm for breaking the security of PRSs.

One way to tackle this is to consider quantum adversaries that can query a classical oracle. If we can show that PRSs can be broken by a polynomial-time quantum algorithm with oracle access to some language $\Lang : \{0,1\}^* \to \{0,1\}$, we conclude that if PRSs exist, then $\Lang \not\in \mathsf{BQP}$. A priori, it is not immediately obvious whether oracle access to \textit{any} language $\Lang$ suffices for a polynomial-time quantum adversary to break PRSs. For our first result, we show that a $\mathsf{PP}$-complete language works. Hence, if $\mathsf{BQP} = \mathsf{PP}$, then PRSs do not exist.

\begin{theorem}[Informal version of \Cref{thm:pp_oracle}]
\label{thm:pp_informal}
There exists a polynomial-time quantum algorithm augmented with a $\mathsf{PP}$ oracle that can distinguish PRSs from Haar-random states.
\end{theorem}

This raises the natural question of whether the $\mathsf{PP}$ oracle in the above theorem can be made weaker. For instance, can we break PRSs with a $\mathsf{QCMA}$ or $\mathsf{QMA}$ oracle, coinciding with our intuition that the task is solvable by a quantum Merlin-Arthur protocol? In our second result, we show that this intuition is perhaps misguided, as we construct a quantum oracle relative to which such a $\mathsf{QMA}$ reduction is impossible.

\begin{theorem}[Informal version of \Cref{thm:bqp^u=qma^u,thm:pru^u}]
\label{thm:qma_informal}
There exists a quantum oracle $\Oracle$ such that:
\begin{enumerate}[(1)]
\item $\mathsf{BQP}^\Oracle = \mathsf{QMA}^\Oracle$, and
\item PRUs (and hence PRSs) exist relative to $\Oracle$.
\end{enumerate}
\end{theorem}

In fact, our oracle $\Oracle$ also satisfies $\mathsf{PromiseBQP}^\Oracle = \mathsf{PromiseQMA}^\Oracle$, which is stronger. For the sake of clarity, in this introduction we will only state our results in terms of classes of languages (e.g.\ $\mathsf{QMA}$) instead of classes of promise problems (e.g.\ $\mathsf{PromiseQMA}$), unless the distinction matters.

Let us remark how bizarre this theorem appears from a cryptographer's point of view. If $\mathsf{BQP} = \mathsf{QMA}$, then \textit{no} computationally-secure classical cryptographic primitives exist, because such primitives can be broken in $\mathsf{NP}$, which is contained in $\mathsf{QMA}$. So, our construction is a black-box separation between PRUs and \textit{all} nontrivial quantum-secure classical cryptography---a relativized world in which any computationally-secure cryptography must use quantum communication. 
\Cref{thm:qma_informal} thus provides a negative answer (in the quantum black box setting) to a question of Ji, Liu, and Song \cite{JLS18-prs} that asks if quantum-secure one-way functions are necessary for pseudorandom states. 

\Cref{thm:qma_informal} illustrates a stark contrast between quantum and classical cryptography, because the existence of hard problems in $\mathsf{NP}$ is necessary to have classical cryptography that is secure against polynomial-time adversaries. One can view our result as evidence that the same is not necessary for the existence of quantum cryptography; perhaps weaker assumptions suffice. Indeed, because a major goal in cryptography is to build cryptosystems from minimal computational assumptions, \Cref{thm:qma_informal} has served as the primary motivation for many recent works that have built cryptography from pseudorandom states and unitaries \cite{AQY22-prs,MY22-prs,BCQ23-efi,MY22-owq,HMY23-qpke}. Note that these works all appeared after this work was originally published \cite{conf-version}.

\subsection{Application: hyperefficient shadow tomography}
An immediate corollary of our results is a new impossibility result for shadow tomography. Aaronson \cite{Aar18-shadow-tomography} defined the shadow tomography problem as the following estimation task: given copies of an $n$-qubit mixed state $\rho$ and a list of two-outcome measurements $O_1,\ldots,O_M$, estimate $\tr(O_i \rho)$ for each $i$ up to additive error $\eps$. Aaronson showed that, remarkably, this is possible using very few copies of $\rho$: just $\poly(n, \log M, \frac{1}{\eps})$ copies suffice, which is polylogarithmic in both the dimension of $\rho$ and the number of quantities to be estimated.

Aaronson then asked in what cases shadow tomography can be made \textit{computationally} efficient with respect to $n$ and $\log M$. Of course, just writing down the input to the problem would take $\Omega(4^n M)$ time if the measurements are given explicitly as Hermitian matrices, and listing the outputs would also take $\Omega(M)$ time. But perhaps one could hope for an algorithm that only operates \textit{implicitly} on both the inputs and outputs. For example, suppose we stipulate the existence of a quantum algorithm that performs the measurement $O_i$ given input $i \in [M]$, and that this algorithm runs in time $\poly(n,\log M)$. Consider a shadow tomography procedure that takes a description of such an algorithm as input, and that outputs a quantum circuit $C$ such that $|C(i) - \tr(O_i \rho)| \le \eps$ for each $i \in [M]$.\footnote{Note the slight abuse of notation here, as the shadow tomography procedure can err with some small probability, and $C$ itself might be a probabilistic quantum circuit. For simplicity, we assume that the shadow tomography procedure always succeeds and that $C$ is deterministic in this exposition.} Aaronson calls this a ``hyperefficient'' shadow tomography protocol if it additionally runs in time $\poly(n, \log M, \frac{1}{\eps})$.

Aaronson gave some evidence that hyperefficient shadow tomography is unlikely to exist, by observing that if hyperefficient shadow tomography is possible, then quantum advice can always be efficiently replaced by classical advice---in other words, $\mathsf{BQP/qpoly} = \mathsf{BQP/poly}$. However, Aaronson and Kuperberg \cite{AK07-qcma-qma} showed a quantum oracle $\UOracle$ relative to which $\mathsf{BQP^\UOracle/qpoly} \neq \mathsf{BQP^\UOracle/poly}$, which implies that hyperefficient shadow tomography is impossible if the observables are merely given as a black box that implements the measurement. The proof of this oracle separation amounts to showing that if the oracle $\UOracle$ either (1) implements a reflection about a Haar-random $n$-qubit state, or (2) acts as the identity, then no $\poly(n)$-query algorithm can distinguish these two cases, even given a classical witness of size $\poly(n)$.

One can consider stronger forms of query access to the observables. For instance, in the common scenario where each observable measures fidelity with a pure state, meaning it has the form $O_i = \ket{\psi_i}\bra{\psi_i}$, then in addition to the ability to measure overlap with $\ket{\psi_i}$, one might also have the power to produce copies of $\ket{\psi_i}$. Note that the ability to prepare $\ket{\psi_i}$ is generally much more powerful than the ability to recognize $\ket{\psi_i}$, the latter of which is equivalent to oracle access to the reflection $\identity - 2\ket{\psi_i}\bra{\psi_i}$. For example, Aaronson and Kuperberg's oracle separation of $\mathsf{QCMA}$ and $\mathsf{QMA}$ \cite{AK07-qcma-qma} amounts to building an oracle relative to which certain quantum states can be recognized efficiently but cannot be approximately prepared by small quantum circuits.
Other black-box separations of state preparation and state reflection are known, e.g.\ \cite{BR20-apxcount}, so one might hope that this type of query access could be substantially more powerful for shadow tomography as well.

Nevertheless, our results imply that black-box hyperefficient shadow tomography is impossible even in this setting where we have state preparation access to the observables. This follows from the simple observation that hyperefficient shadow tomography of this form would suffice to break PRS ensembles with a ($\mathsf{Promise}$)$\mathsf{QCMA}$ oracle.

\begin{theorem}
If a hyperefficient shadow tomography procedure exists that works for any list of observables of the form $\ket{\psi_1}\bra{\psi_1}, \ldots,\allowbreak \ket{\psi_M}\bra{\psi_M}$ given state preparation access to $\ket{\psi_1},\ldots,\allowbreak\ket{\psi_M}$, then all PRS ensembles can be broken by polynomial-time quantum adversaries with oracle access to $\mathsf{PromiseQCMA}$.
\end{theorem}

\begin{proof}[Proof sketch]
For a given PRS ensemble $\{\ket{\varphi_k}\}_{k \in \Key}$, we have state preparation access to the observable list $\{\ket{\varphi_k} \bra{\varphi_k}\}_{k \in \Key}$ by way of the generating algorithm of the PRS. Hence, we can run hyperefficient shadow tomography using this observable list on copies of some unknown state $\ket{\psi}$. Suppose that with high probability, this produces a quantum circuit $C$ such that for each $k \in \Key$, $\Pr\left[|C(k) - \tr(\ket{\varphi_k} \braket{\varphi_k|\psi}\bra{\psi})| \le \frac{1}{10}\right] \ge \frac{2}{3}$ . Observe that the problem of deciding whether there exists some $k$ such that $C(k) \ge \frac{9}{10}$ w.h.p.\ is in $\mathsf{PromiseQCMA}$. If $\ket{\psi}$ is pseudorandom, then such a $k$ always exists (whichever $k$ satisfies $\ket{\psi} = \ket{\varphi_k}$), whereas if $\ket{\psi}$ is Haar-random, such a $k$ exists with negligible probability over the choice of $\ket{\psi}$. Hence, these two ensembles can be distinguished by feeding $C$ into this $\mathsf{PromiseQCMA}$ promise problem.
\end{proof}

The above theorem also relativizes, in the sense that if the shadow tomography procedure only accesses the state preparation algorithm via a black box $\Oracle$, then hyperefficient shadow tomography lets us break PRSs in polynomial time with oracle access to $\Oracle$ and $\mathsf{PromiseQCMA}^\Oracle$. Since \Cref{thm:qma_informal} gives an oracle relative to which $\mathsf{PromiseBQP}^\Oracle = \mathsf{PromiseQCMA}^\Oracle = \mathsf{PromiseQMA}^\Oracle$ and PRSs exist, we conclude that hyperefficient shadow tomography is impossible with only black-box state preparation access to the observables.

\subsection{Our techniques}

We briefly summarize the proof techniques used in our main results.

\subsubsection{Approximate \texorpdfstring{$t$}{t}-designs}
Approximate $t$-designs play a role in the proof of \Cref{thm:pp_informal}. So, in \Cref{sec:designs}, we give formal definitions of $t$-designs and prove some of their useful properties in the context of quantum query complexity. In particular, we establish conditions under which which substituting the Haar measure with a $t$-design yields a relative-error approximation to the acceptance probability of a quantum query algorithm. Several authors have implicitly assumed without proof that this property holds, e.g.\ \cite{BHH16-designs,AMR20-random}, and also an earlier version of this work \cite{conf-version}. We consider it valuable to place these results on more rigorous footing, and believe that the results about $t$-designs proved herein could find independent uses in other complexity-theoretic contexts.

\subsubsection{Breaking pseudorandomness with \texorpdfstring{$\mathsf{PP}$}{PP}}
The starting point for the proof of \Cref{thm:pp_informal}, which gives an upper bound of $\mathsf{PP}$ on the power needed to break pseudorandom states, is a theorem of Huang, Kueng, and Preskill \cite{HKP20-classical-shadows} that gives a simple procedure (sometimes called the \textit{classical shadows algorithm}) for shadow tomography.

\begin{theorem}[Classical shadows \cite{HKP20-classical-shadows}]
\label{thm:HKP}
Fix $M$ different observables $O_1, O_2, \ldots, O_M$ and an unknown $n$-qubit mixed state $\rho$. Then there exists a quantum algorithm that performs $T = O(\log(M / \delta) / \eps^2 \cdot \max_i \tr(O_i^2))$ single-copy measurements in random Clifford bases\footnote{Recall that the Clifford group is the group of unitary transformations generated by Hadamard, phase, and CNOT gates. A Clifford basis is any basis that can be obtained from the computational basis via multiplication by an element of the Clifford group.} of $\rho$, and uses the measurement results (called \emph{classical shadows}) to estimate the quantities $\tr(O_1 \rho), \tr(O_2 \rho), \ldots, \tr(O_M \rho)$, such that with probability at least $1 - \delta$, all of the $M$ quantities are correct up to additive error $\eps$.
\end{theorem}

If $\{\ket{\varphi_k}\}_{k \in \Key}$ is a pseudorandom state ensemble, then by choosing $O_k = \ket{\varphi_k}\bra{\varphi_k}$ for each key $k \in \Key$ to be the list of observables, we can use the above algorithm to determine whether $\rho$ is close to one of the states in the PRS ensemble. A Haar-random state will be far from \textit{all} of the pseudorandom states with overwhelming probability. Hence, \Cref{thm:HKP} implies the existence of an algorithm that distinguishes the pseudorandom and Haar-random ensembles, by performing a polynomial number of random Clifford measurements and analyzing the results. The key observation is that the Clifford measurements can be performed efficiently, even though the resulting analysis (which operates on purely classical information) might be computationally expensive.

Next, one could try to argue that the computationally difficult steps in the above algorithm can be made efficient with a $\mathsf{PP}$ oracle. However, we take a different approach. We adopt a Bayesian perspective: suppose that with $50\%$ probability we are given copies of a Haar-random state, and otherwise with $50\%$ probability we are given copies of a randomly chosen state from the pseudorandom ensemble. We wish to distinguish these two cases using only the results of the random Clifford measurements as observed data. One way to do this is via the Bayes decision rule: we compute the posterior probability of being Haar-random or pseudorandom given the measurements, and then guess the more likely result. In fact, the Bayes decision rule is well-known to be the \textit{optimal} decision rule in general, in the sense that any decision rule errs at least as often as the Bayes decision rule (see e.g.\ \cite[Chapter 4.4.1]{Ber13-bayes-decision}). Hence, because the algorithm of Huang, Kueng, and Preskill (\Cref{thm:HKP}) distinguishes the Haar-random and pseudorandom ensembles with good probability, the Bayes decision rule conditioned on the random Clifford measurements must work \textit{at least} as well at the same distinguishing task.

Finally, we observe that using a quantum algorithm with postselection, we can approximate the relevant posterior probabilities needed for the Bayes decision rule. This allows us to appeal to the equivalence $\mathsf{PostBQP} = \mathsf{PP}$ \cite{Aar05-postbqp} to simulate this postselection with a $\mathsf{PP}$ oracle.

Technically, one challenge is that the postselected quantum algorithm requires the ability to prepare copies of a Haar-random state, even though a polynomial-time quantum algorithm cannot even approximately prepare most Haar-random states. The solution is to replace the Haar ensemble by an approximate quantum design, which we argue does not substantially change the success probability of the algorithm. 

\subsubsection{Instantiating pseudorandomness with \texorpdfstring{$\mathsf{BQP} = \mathsf{QMA}$}{BQP = QMA}}
For our second result (\Cref{thm:qma_informal}), the oracle $\Oracle$ that we construct consists of two parts: a quantum oracle $\UOracle = \{\UOracle_n \}_{n \in \Naturals}$, where each $\UOracle_n$ consists of $2^n$ different Haar-random $n$-qubit unitary matrices, and a classical oracle (i.e.\ a language) $\COracle$ that we build independently of $\UOracle$. We prove that \Cref{thm:qma_informal} holds with probability $1$ over the choice of $\UOracle$.

Showing that PRUs exist relative to $(\UOracle, \COracle)$ is reasonably straightforward. Notably, the security proof does not depend on the choice of $\COracle$, so long as $\COracle$ is independent of the randomly sampled $\UOracle$. The proof uses the BBBV theorem (i.e.\ the optimality of Grover's algorithm) \cite{BBBV97-search}, and is analogous to showing that one-way functions or pseudorandom generators exist relative to a random \textit{classical} oracle, as was shown by Impagliazzo and Rudich \cite{IR89-permutations}. We only rigorously prove security against adversaries with classical advice, though we believe that the framework of Chung, Guo, Liu, and Qian \cite{CGLQ20-tradeoffs} should yield a security proof against adversaries with quantum advice.

Slightly more technically involved is proving that $\mathsf{BQP}^{\UOracle, \COracle} = \mathsf{QMA}^{\UOracle, \COracle}$. To do so, we argue that a $\mathsf{QMA}$ verifier is not substantially more powerful than a $\mathsf{BQP}$ machine at learning nontrivial properties of $\UOracle$. More precisely, we argue that if a $\mathsf{QMA}$ verifier $\Ver$ makes $T$ queries to $\UOracle_n$ for some $n \in \Naturals$, then either (1) $n = O(\log T)$ is sufficiently small that $\poly(T)$ queries to $\UOracle_n$ actually suffice to learn $\UOracle_n$ to inverse-polynomial precision, or else (2) $n = \omega(\log T)$ is sufficiently large that with high probability, the maximum acceptance probability of $\Ver$ (over the choice of Merlin's witness) is close to the average maximum acceptance probability of $\Ver$ when $\UOracle_n$ is replaced by a random set of matrices sampled from the Haar measure. We prove this as a consequence of the extremely strong concentration of measure properties exhibited by the Haar measure \cite{Mec19-random-matrix}.

For a certain carefully-constructed language $\COracle$, this allows a $\mathsf{BQP}^{\UOracle, \COracle}$ machine to approximate the maximum acceptance probability of $\Ver^{\UOracle, \COracle}$ as follows. In case (1), the $\mathsf{BQP}^{\UOracle, \COracle}$ machine first queries $\UOracle_n$ enough times to learn a unitary transformation $\widetilde{\UOracle}_n$ that is close to $\UOracle_n$, and then hard codes $\widetilde{\UOracle}_n$ into a new $\mathsf{QMA}^\COracle$ verifier $\VerW^{\COracle}$ that simulates $\Ver$ by replacing queries to $\UOracle_n$ with calls to $\widetilde{\UOracle_n}$. In case (2), the $\mathsf{BQP}^{\UOracle, \COracle}$ machine similarly constructs a new $\mathsf{QMA}^\COracle$ verifier $\VerW^{\COracle}$, instead simulating $\Ver$ by replacing queries to $\UOracle_n$ with queries to \textit{independently chosen} Haar-random unitaries $\overline{\UOracle}_n$. Thus, the problem of approximating the maximum acceptance probability of $\Ver^{\UOracle, \COracle}$ reduces to approximating the maximum acceptance probability of $\VerW^{\COracle}$, averaged over $\overline{\UOracle}_n$. The language $\COracle$ is constructed in such a fashion that querying $\COracle$ on a description of $\VerW$ returns the desired approximation.

\subsection{Open problems}
Can we prove a similar result to \Cref{thm:qma_informal} using a \textit{classical} oracle, for either PRUs or PRSs? Attempting to resolve this question seems to run into many of the same difficulties that arise in constructing a classical oracle separation between $\mathsf{QCMA}$ and $\mathsf{QMA}$, which also remains an open problem \cite{AK07-qcma-qma}. For one, as pointed out in \cite{AK07-qcma-qma}, we do not even know whether every $n$-qubit unitary transformation can be approximately implemented in $\poly(n)$ time relative to some classical oracle---this is sometimes known as the \textit{unitary synthesis problem} \cite{Aar16-barbados,Ros21-unitary,LMW24-synthesis}. Even if one could resolve this, it is not clear whether the resulting PRUs or PRSs would be secure against adversaries with the power of $\mathsf{QMA}$. For instance, we show in \Cref{app:binary_phase_prs} that an existing construction of PRSs, whose security is provable in the random oracle model \cite{BS19-binary}, can be broken with an $\mathsf{NP}$ oracle. Nevertheless, recent work by Kretschmer, Qian, Sinha, and Tal \cite{KQST23-prs} makes progress on this question by constructing an oracle relative to which $\mathsf{P} = \mathsf{NP}$ and a weaker version of pseudorandom states (with only single-copy security) exist.

What else can be said about the hardness of learning quantum states and unitary transformations, either in the worst case or on average? A related question is to explore the hardness of problems involving quantum \textit{meta-complexity}: that is, problems that themselves encode computational complexity or difficulty. Consider, for example, a version of the minimum circuit size problem ($\mathsf{MCSP}$) for quantum states: given copies of a pure quantum state $\ket{\psi}$, determine the size of the smallest quantum circuit that approximately outputs $\ket{\psi}$. If PRSs exist, then this task should be hard, but placing an upper bound on the complexity of this task might be difficult in light of our results. We view this problem as particularly intriguing because it does not appear to have an obvious classical analogue, and also because of its relevance to the wormhole growth paradox and Susskind's Complexity=Volume conjecture in AdS/CFT \cite{BFV20-prs-wormhole,Sus16-horizons,Sus16-addendum}. A number of recent breakthroughs in complexity theory have involved ideas from meta-complexity (see surveys by Allender \cite{All17-meta,All20-mcsp} or Lu and Oliveira \cite{LO22-survey}), and it would be interesting to see which of these techniques could be ported to the quantum setting.

What other complexity-theoretic evidence can be given for the existence of PRSs and PRUs? Can we give candidate constructions of PRSs or PRUs that do not rely on the assumption $\mathsf{BQP} \neq \mathsf{QMA}$? To give a specific example, an interesting question is whether polynomial-size quantum circuits with random local gates form PRUs. Random circuits are known to information-theoretically approximate the Haar measure in the sense that they form approximate unitary designs \cite{BHH16-designs}, and it seems conceivable that they could also be computationally pseudorandom.

\subsection{Conference version}
This paper improves upon the earlier conference version \cite{conf-version} in two major ways. First, the section on $t$-designs (\Cref{sec:designs}) is a new addition. Second, the oracle $\Oracle = (\UOracle, \COracle)$ that we use to collapse $\mathsf{PromiseQMA}^{\UOracle, \COracle}$ to $\mathsf{PromiseBQP}^{\UOracle, \COracle}$ is different. In both versions, the quantum part $\UOracle$ is the same. However, the conference version claimed that the classical oracle $\COracle$ could be any $\mathsf{PSPACE}$-complete language. By contrast, in this paper $\COracle$ is a specific recursively-constructed language. The reason for this change is an error that was discovered in the proof. We discuss this error, along with the prospects of restoring the original oracle construction, in \Cref{sec:alternative_oracles}.

\section{Preliminaries}

\subsection{Basic notation}

Throughout, $[n]$ denotes the set of integers $\{1,2,\ldots,n\}$. If $x \in \{0,1\}^n$ is a binary string, then $|x|$ denotes the length of $x$, and $\wt(x)$ denotes its Hamming weight. For $X$ a finite set, we let $|X|$ denote the size of $X$. If $X$ is a probability distribution, then we use $x \sim X$ to denote a random variable $x$ sampled according to $X$. When $X$ is a finite set, we also use $x \sim X$ to indicate a random variable $x$ drawn uniformly from $X$. A function $f(n)$ is \textit{negligible} if for every constant $c > 0$, $f(n) \le \frac{1}{n^c}$ for all sufficiently large $n$. We use $\negl(n)$ to denote an arbitrary negligible function, and $\poly(n)$ to denote an arbitrary polynomially-bounded function.

\subsection{Probability}

We require two basic facts about probability. The first regards the optimality of the Bayes decision rule, which is a strategy for guessing a random variable $X$ from posterior information $Y$. The Bayes decision rule is to always guess the value of $X$ that maximizes the posterior probability given $Y$. The Bayes decision rule is optimal, in the sense that any other strategy that guesses $X$ using $Y$ errs at least as often as the Bayes decision rule. We only need the following special case of this fact, which also applies more generally (see \cite[Chapter 4.4.1]{Ber13-bayes-decision} for further discussion).

\begin{lemma}[Bayes decision rule]
    \label{lem:bayes_decision}
    Let $X$ be a $\{0,1\}$-valued random variable, let $Y$ be a random variable (not necessarily independent of $X$) with domain $D$, and let $f: D \to \{0,1\}$. Then:
    \[
    \Pr[f(Y) = X] \le \Pr\left[\argmax_i \Pr[X = i | Y] = X\right].
    \]
\end{lemma}

The other fact we need is the Borel--Cantelli lemma for sequences of probabilistic events. It gives a criterion under which at most finitely many of the events occur, with probability $1$.

\begin{lemma}[Borel--Cantelli \cite{Bor09-prob,Can17-prob}]
\label{lem:borel-cantelli}
    Let $\{X_n\}_{n \in \Naturals}$ be a sequence of (not necessarily independent) random variables with values in $\{0, 1\}$. If
    \[
    \sum_{n=1}^\infty \E[X_n] < \infty,
    \]
    then
    \[
    \Pr\left[\sum_{n=1}^\infty X_n = \infty\right] = 0.
    \]
\end{lemma}

\subsection{Quantum information}
We let $\TD(\rho, \sigma)$ denote the trace distance between density matrices $\rho$ and $\sigma$. For a matrix $M$ we use $||M||_F \coloneqq \sqrt{\tr \left(M^\dagger M\right)}$ to denote its Frobenius norm.

When $A$ and $B$ are Hermitian matrices, we use $A \preceq B$ to denote the semidefinite ordering, i.e.\ that $B - A$ is positive semidefinite. We extend this notation to superoperators $A$ and $B$: $A \preceq B$ denotes that $B - A$ is \textit{completely positive}. A superoperator $\Lambda$ is said to be completely positive if, for any identity superoperator $\identity$ and positive semidefinite matrix $\rho$, $(\Lambda \otimes \identity)(\rho)$ is positive semidefinite. If $\Lambda$ has input dimension $N$, a criterion equivalent to complete positivity is
\[
(\Lambda \otimes \identity_N)(\ket{\Phi_N}\bra{\Phi_N}) \succeq 0,
\]
where $\identity_N$ is the $N$-dimensional identity channel, and $\ket{\Phi_N} \coloneqq \frac{1}{\sqrt{N}} \sum_{i=1}^N \ket{i}\ket{i}$ is the standard maximally entangled state of dimension $N \times N$ \cite{BHH16-designs}.

For a unitary matrix $U$, we use $U \cdot U^\dagger$ to denote the superoperator that maps a density matrix $\rho$ to $U\rho U^\dagger$. In a slight abuse of notation, if $\Alg$ denotes a quantum algorithm (which may consist of unitary gates, measurements, oracle queries, and initialization of ancilla qubits), then we also use $\Alg$ to denote the superoperator corresponding to the action of $\Alg$ on input density matrices.

We let $\diamondnorm{\Alg}$ denote the \textit{diamond norm} \cite{AKN98-mixed} of a superoperator $\Alg$ acting on density matrices, which is defined by
\[
\diamondnorm{\Alg} \coloneqq \sup_{\tr(\rho) = 1, \rho \succeq 0} ||(\Alg \otimes I)(\rho)||_1,
\]
where $I$ denotes the identity superoperator acting on a space of the same dimension as $\Alg$. Intuitively, the diamond norm gives an analogue of trace distance for channels: the distance between two channels in the diamond norm captures the maximum bias by which those two channels can be distinguished. In particular, we have the following:

\begin{fact}
    \label{fact:diamond_bounds_trace}
    Let $\Alg$ and $\AlgB$ be quantum channels and $\rho$ a density matrix. Then
    \[
    \TD(\Alg(\rho), \AlgB(\rho)) \le \frac{1}{2}\diamondnorm{\Alg - \AlgB}.
    \]
\end{fact}

We use the following formula for the distance between unitary superoperators in the diamond norm.
\begin{fact}[\cite{AKN98-mixed}]
\label{fact:diamond_norm}
Let $U$ and $V$ be unitary matrices, and suppose $d$ is the distance between $0$ and the polygon in the complex plane whose vertices are the eigenvalues of $UV^\dagger$. Then
$$\diamondnorm{U \cdot U^\dagger - V \cdot V^\dagger} = 2\sqrt{1 - d^2}.$$
\end{fact}

A consequence of \Cref{fact:diamond_norm} is the following well-known bound relating distance in diamond norm to the Frobenius norm for unitary channels. We provide a proof for completeness.

\begin{lemma}\label{lem:frobenius_diamond}
Let $U, V$ be $N \times N$ unitary matrices. Then $||U \cdot U^\dagger - V \cdot V^\dagger||_\diamond \le 2||U - V||_F$.
\end{lemma}

\begin{proof}
Let $\{\lambda_i : i \in [N]\}$ denote the eigenvalues of $UV^\dagger$. Then we have:
\begin{align}
||U - V||_F^2 &= \tr\left((U - V)(U - V)^\dagger\right)\nonumber\\
&= \tr(2I - UV^\dagger - VU^\dagger)\nonumber\\
&= 2N - \sum_{i=1}^N (\lambda_i + \lambda_i^*)\nonumber\\
&= \sum_{i=1}^N (2 - 2\Re(\lambda_i))\nonumber\\
&\ge \max_i(2 - 2\Re(\lambda_i))\label{eq:frobenius_upper_bound},
\end{align}
where $\Re(\lambda_i)$ denotes the real part of $\lambda_i$. The last line holds because the eigenvalues of a unitary matrix have absolute value $1$. 

Let $d$ be the distance in the complex plane between $0$ and the polygon whose vertices are $\lambda_1,\ldots,\lambda_N$. Then from \Cref{fact:diamond_norm} we may conclude:
\begin{align*}
||U \cdot U^\dagger - V \cdot V^\dagger||_\diamond &\le \max_i 2\sqrt{1 - \max\{\Re(\lambda_i),0\}^2}\\
&\le \max_i 2\sqrt{2 - 2\Re(\lambda_i)}\\
&\le 2||U - V||_F,
\end{align*}
where the first inequality uses the fact that either all of the eigenvalues have positive real components and therefore $d \ge \min_i \Re(\lambda_i)$, or else $d \ge 0$; the second inequality substitutes $1 - \max\{x, 0\}^2 \le 2 - 2x$ which holds for all $x \in \Reals$; and the third inequality substitutes \eqref{eq:frobenius_upper_bound}.
\end{proof}

\subsection{Haar measure and concentration}

We use $\States(N)$ to denote the set of $N$-dimensional pure quantum states, and $\Unitaries(N)$ to denote the group of $N \times N$ unitary matrices. When $N = 2^n$, we identify these with $n$-qubit states and unitary transformations, respectively. We use $\sigma_N$ to denote the Haar measure on $\States(N)$, and we let $\mu_N$ denote the Haar measure over $\Unitaries(N)$. We write $\Unitaries(N)^M$ for the space of $MN \times MN$ block-diagonal unitary matrices, where each block has size $N \times N$, and we also identify $\Unitaries(N)^M$ with $M$-tuples of $N \times N$ unitary matrices. We use $\mu_N^M$ to denote the product measure $\mu_N^M(U_1,U_2,\ldots,U_M) \coloneqq \mu_N(U_1) \cdot \mu_N(U_2) \cdots \mu_N(U_M)$ on $\Unitaries(N)^M$, which we interpret as a distribution over a direct sum $U_1 \oplus U_2 \oplus \ldots \oplus U_M$ of matrices.

We require the following concentration inequality on the Haar measure, which is stated in terms of Lipschitz continuous functions. For a metric space $\mathcal{M}$ with metric $d$, a function $f: \mathcal{M} \to \Reals$ is \textit{$L$-Lipschitz} if for all $x, y \in \mathcal{M}$, $|f(x) - f(y)| \le L \cdot d(x, y)$.

\begin{theorem}[{\cite[Theorem 5.17]{Mec19-random-matrix}}]
\label{thm:haar_concentration}
Given $N_1,\ldots,N_k \in \Naturals$, let $X = \Unitaries(N_1) \oplus \cdots \oplus \Unitaries(N_k)$ be the space of block-diagonal unitary matrices with blocks of size $N_1, \ldots, N_k$. Let $\mu = \mu_{N_1} \times \cdots \times \mu_{N_k}$ be the product of Haar measures on $X$. Suppose that $f: X \to \Reals$ is $L$-Lipschitz in the Frobenius norm. Then for every $t > 0$:
$$\Pr_{U\sim \mu}\left[f(U) \ge \E_{V \sim \mu}[f(V)] + t\right] \le \exp\left(-\frac{(N-2)t^2}{24L^2}\right),$$
where $N = \min\{N_1,\ldots,N_k\}$.
\end{theorem}

\subsection{Complexity theory}

A \textit{language} is a function $L: \{0,1\}^* \to \{0,1\}$. A \textit{promise problem} is a function $\Pi: \{0,1\}^* \to \{0,1,\bot\}$. The \textit{domain} of a promise problem $\Pi$, denoted $\Dom(\Pi)$, is
\[\Dom(\Pi) \coloneqq \left\{ x \in \{0,1\}^* : \Pi(x) \in \{0,1\} \right\}\]

We assume familiarity with standard complexity classes such as $\mathsf{BQP}$ and $\mathsf{PP}$. For completeness, we define some complexity classes used prominently in this work.

\begin{definition}
\label{def:BPP}
\label{def:BQP}
A promise problem $\Pi : \{0,1\}^* \to \{0,1,\bot\}$ is in $\mathsf{PromiseBQP}$ (\underline{B}ounded-error \underline{Q}uantum \underline{P}olynomial time) if there exists a randomized polynomial-time quantum algorithm $\Alg(x)$ such that:
\begin{enumerate}[(i)]
\item If $\Pi(x) = 1$, then $\Pr\left[\Alg(x) = 1\right] \ge \frac{2}{3}$.
\item If $\Pi(x) = 0$, then $\Pr\left[\Alg(x) = 1\right] \le \frac{1}{3}$.
\end{enumerate}
$\mathsf{BQP}$ is defined as the set of languages in $\mathsf{PromiseBQP}$.
\end{definition}

\begin{definition}
\label{def:QMA}
A promise problem $\Pi : \{0,1\}^* \to \{0,1,\bot\}$ is in $\mathsf{PromiseQMA}$ (\underline{Q}uantum \underline{M}erlin--\underline{A}rthur) if there exists a polynomial-time quantum algorithm $\Ver(x,\ket{\psi})$ called a \emph{verifier} and a polynomial $p$ such that:
\begin{enumerate}[(i)]
\item (Completeness) If $\Pi(x) = 1$, then there exists a quantum state $\ket{\psi}$ on $p(|x|)$ qubits (called a \emph{witness} or \emph{proof}) such that $\Pr\left[\Ver(x,\ket{\psi}) = 1\right] \ge \frac{2}{3}$.
\item (Soundness) If $\Pi(x) = 0$, then for every state $\ket{\psi}$ on $p(|x|)$ qubits, $\Pr\left[\Ver(x,\ket{\psi}) = 1\right] \le \frac{1}{3}$.
\end{enumerate}
$\mathsf{QMA}$ is defined as the set of languages in $\mathsf{PromiseQMA}$.
\end{definition}

Note: we will sometimes call any algorithm of the form $\Ver(x,\ket{\psi})$ a \emph{$\mathsf{QMA}$ verifier}, even if it does not satisfy the promise of a $\mathsf{QMA}$ language.

\begin{definition}
\label{def:PostBQP}
A promise problem $\Pi : \{0,1\}^* \to \{0,1,\bot\}$ is in $\mathsf{PromisePostBQP}$ (\underline{Post}selected \underline{B}ounded-error \underline{Q}uantum \underline{P}olynomial time) if there exists a poly\-no\-mi\-al-time quantum algorithm $\Alg(x)$ that outputs a trit in $\{0,1,*\}$ such that:
\begin{enumerate}[(i)]
\item For all $x \in \Dom(\Pi)$, $\Pr[\Alg(x) \in \{0,1\}] > 0$. When $\Alg(x) \in \{0,1\}$, we say that \emph{postselection succeeds}.
\item If $\Pi(x) = 1$, then $\Pr\left[\Alg(x) = 1 \mid \Alg(x) \in \{0,1\}\right] \ge \frac{2}{3}$. In other words, conditioned on postselection succeeding, $\Alg$ outputs $1$ with at least $\frac{2}{3}$ probability.
\item If $\Pi(x) = 0$, then $\Pr\left[\Alg(x) = 1 \mid \Alg(x) \in \{0,1\}\right] \le \frac{1}{3}$. In other words, conditioned on postselection succeeding, $\Alg$ outputs $1$ with at most $\frac{1}{3}$ probability.
\end{enumerate}
$\mathsf{PostBQP}$ is defined as the set of languages in $\mathsf{PromisePostBQP}$.
\end{definition}

Technically, the definition of $\mathsf{PromisePostBQP}$ is sensitive to the choice of universal gate set used to specify quantum algorithms, as was observed by Kuperberg \cite{Kup15-jones}. However, for most ``reasonable'' gate sets, such as unitary gates with algebraic entries \cite{Kup15-jones}, the choice of gate set is irrelevant. We assume such a gate set, e.g.\ $\{\mathrm{CNOT}, H, T\}$.

We require the following equivalent characterization of $\mathsf{PromisePostBQP}$:

\begin{lemma}[Aaronson \cite{Aar05-postbqp}]
\label{lem:postbqp-pp}
    $\mathsf{PromisePostBQP} = \mathsf{PromisePP}$
\end{lemma}

\subsection{Quantum oracles}
\label{subsec:quantum-oracles}

We frequently consider quantum algorithms that query quantum oracles. In this work, unless otherwise specified, we define queries to a unitary matrix $\UOracle$ to mean a single application of either $\UOracle$, $\UOracle^\dagger$, controlled-$\UOracle$ (i.e.\ $I \oplus \UOracle$, where $I$ is the identity of the same dimension), or controlled-$\UOracle^\dagger$ (i.e.\ $I \oplus \UOracle^\dagger$), unless otherwise specified.
We use superscript notation for algorithms that query oracles. For instance, $\Alg^\UOracle(x, \ket{\psi})$ denotes a quantum algorithm $\Alg$ that queries an oracle $\UOracle$ and receives a classical input $x$ and a quantum input $\ket{\psi}$.

We consider versions of $\mathsf{PromiseBQP}$, $\mathsf{PromiseQMA}$, and $\mathsf{PromisePostBQP}$ augmented with quantum oracles, where the algorithm (or in the case of $\mathsf{PromiseQMA}$, the verifier) can apply unitary transformations from an infinite sequence $\UOracle = \{\UOracle_n\}_{n \in \Naturals}$. We denote the respective complexity classes by $\mathsf{PromiseBQP}^\UOracle$, $\mathsf{PromiseQMA}^\UOracle$, and $\mathsf{PromisePostBQP}^\UOracle$. We assume the algorithm incurs a cost of $n$ to query $\UOracle_n$ so that a polynomial-time algorithm on input $x$ can query $\UOracle_n$ for any $n \le \poly(|x|)$. In this model, a query to $\UOracle_n$ consists of a single application of either $\UOracle_n$, controlled-$\UOracle_n$, or their inverses.

The quantum oracle model includes classical oracles as a special case. For a language $\Lang$, a query to $\Lang$ is implemented via the unitary transformation $\UOracle$ that acts as $\UOracle\ket{x}\ket{b} = \ket{x}\ket{b \oplus \Lang(x)}$.

\subsection{Cryptography}

We use the following definitions of pseudorandom quantum states (PRSs) and pseudorandom unitaries (PRUs), which were introduced by Ji, Liu, and Song \cite{JLS18-prs}. 

\begin{definition}[Pseudorandom quantum states \cite{JLS18-prs}]
\label{def:prs}
Let $\kappa \in \Naturals$ be the security parameter, and let $n(\kappa)$ be the number of qubits in the quantum system. A keyed family of $n$-qubit quantum states $\{\ket{\varphi_k}\}_{k \in \Key}$ is \emph{pseudorandom} if the following two conditions hold:
\begin{enumerate}[(1)]
\item (Efficient generation) There is a polynomial-time quantum algorithm $G$ that generates $\ket{\varphi_k}$ on input $k$, meaning $G(k) = \ket{\varphi_k}$.
\item (Computationally indistinguishable) For any polynomial-time quantum adversary $\Alg$ and for every $T = \poly(\kappa)$:
$$\left| \Pr_{k \sim \Key}\left[\Alg\left(1^\kappa, \ket{\varphi_k}^{\otimes T}\right) = 1 \right] - \Pr_{\ket{\psi} \sim \sigma_{2^n}}\left[\Alg\left(1^\kappa, \ket{\psi}^{\otimes T}\right) = 1 \right] \right| \le \negl(\kappa).$$
\end{enumerate}
\end{definition}

We emphasize that the above security definition must hold for \textit{all} polynomial values of $T$ (i.e.\ $T$ is not bounded in advance). That being said, there do exist alternative definitions of pseudorandom states in which the adversary only receives a single copy of the state \cite{MY22-prs}.

\begin{definition}[Pseudorandom unitary transformations \cite{JLS18-prs}]
\label{def:pru}
Let $\kappa \in \Naturals$ be the security parameter, and let $n(\kappa)$ be the number of qubits in the quantum system. A keyed family of $n$-qubit unitary transformations $\{U_k\}_{k \in \Key}$ is \emph{pseudorandom} if the following two conditions hold:
\begin{enumerate}[(1)]
\item (Efficient computation) There is a polynomial-time quantum algorithm $G$ that implements $U_k$ on input $k$, meaning that for any $n$-qubit $\ket{\psi}$, $G(k, \ket{\psi}) = U_k\ket{\psi}$.
\item (Computationally indistinguishable) For any polynomial-time quantum algorithm $\Alg^U$ that queries $n$-qubit $U$:
$$\left| \Pr_{k \sim \Key}\left[\Alg^{U_k}\left(1^\kappa\right) = 1 \right] - \Pr_{U \sim \mu_{2^n}}\left[\Alg^{U}\left(1^\kappa\right) = 1 \right] \right| \le \negl(\kappa).$$
\end{enumerate}
\end{definition}

We sometimes call the negligible quantities in the above definitions the \textit{advantage} of the quantum adversary $\Alg$. Additionally, we may instantiate these primitives \textit{relative to an oracle $\Oracle$}, which just means that both the generating algorithm $G$ and the adversary $\Alg$ are additionally allowed to query $\Oracle$.

In this work, we will only consider pseudorandom state and unitary ensembles where $n(\kappa) = \omega(\log \kappa)$. Although the original definition of Ji, Liu, and Song \cite{JLS18-prs} did not impose this condition, later works have shown that $O(\log \kappa)$-qubit pseudorandom ensembles behave very differently from $\omega(\log \kappa)$-qubit ensembles \cite{BS20-scalable,AQY22-prs,BEM24-shrink}; the former are often called ``short PRSs/PRUs''. Intuitively, this difference is because one can perform tomography on a quantum state or unitary of $O(\log \kappa)$ qubits to any desired precision $\eps$ in time $\poly(\kappa, \eps)$, so short PRSs/PRUs behave more like cryptographic objects with classical output.

We must also be careful about the type of adversary $\mathcal{A}$ considered in \Cref{def:prs,def:pru}. In this work, we consider security against non-uniform quantum algorithms with classical advice, which means that the adversary is allowed to be a different polynomial-time quantum algorithm for each setting of the security parameter $\kappa \in \Naturals$. Without loss of generality, such an adversary can always be assumed to take the form of a \textit{uniform} $\poly(\kappa)$-time quantum algorithm $\Alg\left(1^\kappa, x\right)$, where $x \in \{0,1\}^{\poly(\kappa)}$ is an advice string that depends only on $\kappa$.

\section{Approximate \texorpdfstring{$t$}{t}-designs}
\label{sec:designs}

We start by defining an $\eps$-approximate quantum (state) $t$-design, which is a distribution over quantum states that information-theoretically approximates the Haar measure over states. While there are several definitions of approximate $t$-designs used in the literature, for this work it is crucial that we use \textit{multiplicative} approximate designs for both states and unitaries, meaning that the designs approximate the first $t$ moments of the Haar measure to within a multiplicative $1 \pm \eps$ error (as opposed to additive error).

\begin{definition}[Approximate quantum design, cf.\ \cite{AE07-designs}]
\label{def:state_t_design}
A probability distribution $S$ over $\States(N)$ is an \emph{$\eps$-approximate quantum $t$-design} if:
$$(1 - \eps)\E_{\ket{\psi} \sim \sigma_N} \left[\ket{\psi}\bra{\psi}^{\otimes t}\right] \preceq \E_{\ket{\psi} \sim S} \left[\ket{\psi}\bra{\psi}^{\otimes t}\right] \preceq (1 + \eps)\E_{\ket{\psi} \sim \sigma_N} \left[\ket{\psi}\bra{\psi}^{\otimes t}\right].$$
\end{definition}

Similarly, we require $\eps$-approximate \textit{unitary} $t$-designs, which are approximations to the Haar measure over unitary matrices. 

\begin{definition}[Approximate unitary design \cite{BHH16-designs}]
\label{def:unitary_t_design}
A probability distribution $S$ over $\Unitaries(N)$ is an \emph{$\eps$-approximate unitary $t$-design} if:
$$(1 - \eps)\E_{U \sim \mu_N} \left[(U \cdot U^\dagger)^{\otimes t}\right] \preceq \E_{U \sim S} \left[(U \cdot U^\dagger)^{\otimes t}\right] \preceq (1 + \eps)\E_{U \sim \mu_N} \left[(U \cdot U^\dagger)^{\otimes t}\right].$$
\end{definition}

An important observation is that unitary designs give rise to state designs:

\begin{proposition}
\label{prop:unitary_to_state_design}
Let $S$ be an $\eps$-approximate unitary $t$-design over $\Unitaries(N)$. Then for any $\ket{\varphi} \in \States(N)$, $S\ket{\varphi}$ is an $\eps$-approximate quantum $t$-design.
\end{proposition}

\begin{proof}
We only establish the right inequality in \Cref{def:state_t_design}; the proof of the left inequality is similar. We have:
\begin{align*}
\E_{\ket{\psi} \sim S\ket{\varphi}} \left[\ket{\psi}\bra{\psi}^{\otimes t}\right]
&= 
\E_{U \sim S} \left[U^{\otimes t}
\left(\ket{\varphi}\bra{\varphi}^{\otimes t}\right) (U^\dagger)^{\otimes t}\right]\\
&\preceq (1 + \eps)\E_{U \sim \mu_N} \left[U^{\otimes t}
\left(\ket{\varphi}\bra{\varphi}^{\otimes t}\right) (U^\dagger)^{\otimes t}\right]\\
&= (1 + \eps)\E_{\ket{\psi} \sim \sigma_N} \left[\ket{\psi}\bra{\psi}^{\otimes t}\right],
\end{align*}
where the second line applies \Cref{def:unitary_t_design} and the definition of complete positivity, and the last line uses the invariance of the Haar measure. This implies that $S\ket{\varphi}$ satisfies \Cref{def:state_t_design}.
\end{proof}

Efficient constructions of approximate unitary $t$-designs over qubits are known, as below.

\begin{lemma}
\label{lem:efficient_unitary_designs}
For each $n, t \in \Naturals$ and $\eps > 0$, there exists $m \le \poly(n,t,\log \frac{1}{\eps})$ and a $\poly(n,t,\log \frac{1}{\eps})$-time classical algorithm $\mathcal{S}$ that takes as input a random string $x \sim \{0,1\}^m$ and outputs a description of a quantum circuit on $n$ qubits such that the circuits sampled from $\mathcal{S}$ form an $\eps$-approximate unitary $t$-design over $\Unitaries(2^n)$.
\end{lemma}

\begin{proof}[Proof sketch]
Fix an arbitrary universal quantum gate set $G$ with algebraic entries that is closed under taking inverses (e.g.\ $G = \{\mathrm{CNOT}, H, T, T^\dagger\}$). Brand\~{a}o, Harrow, and Horodecki \cite[Corollary 7]{BHH16-designs} show that $n$-qubit quantum circuits consisting of $\poly(n, t, \log \frac{1}{\eps})$ random gates sampled from $G$, applied to random pairs of qubits, form $\eps$-approximate unitary $t$-designs. So, $\mathcal{S}$ just has to sample from this distribution, which can be done with $\poly(n, t, \log \frac{1}{\eps})$ bits of randomness.
\end{proof}

Note that this also implies an efficient construction of $\eps$-approximate quantum (state) $t$-designs, by taking $\ket{\varphi} = \ket{0^n}$ in \Cref{prop:unitary_to_state_design}.

Essentially the only property we need of approximate $t$-designs is that they can be used in place of the Haar measure in any quantum algorithm that uses $t$ copies of a Haar-random state (or $t$ queries to a Haar-random unitary), and the measurement probabilities of the algorithm will change by only a small multiplicative factor.

\begin{lemma}
\label{lem:state_t_design_multiplicative}
Let $S$ be an $\eps$-approximate quantum $t$-design over $\States(N)$, and let $\Alg$ be an arbitrary quantum measurement. Then:
$$(1 - \eps)\Pr_{\ket{\psi} \sim \sigma_N}\left[\Alg\left(\ket{\psi}^{\otimes t}\right) = 1\right] \le \Pr_{\ket{\psi} \sim S}\left[\Alg\left(\ket{\psi}^{\otimes t}\right) = 1\right] \le (1 + \eps) \Pr_{\ket{\psi} \sim \sigma_N}\left[\Alg\left(\ket{\psi}^{\otimes t}\right) = 1\right].$$
\end{lemma}

\begin{proof}
Let $0 \preceq M \preceq \identity$ be the measurement performed at the end of the algorithm, so that on input a state $\rho$, $\Pr[\Alg(\rho) = 1] = \tr(M \rho)$. Then:
\begin{align*}
\Pr_{\ket{\psi} \sim S}\left[\Alg\left(\ket{\psi}^{\otimes t}\right) = 1\right] &=
\tr\left(M \E_{\ket{\psi} \sim S} \left[\ket{\psi}\bra{\psi}^{\otimes t}\right]\right)\\
&\le \tr\left(M (1 + \eps) \E_{\ket{\psi} \sim \mu_N} \left[\ket{\psi}\bra{\psi}^{\otimes t}\right]\right)\\
&= (1 + \eps)\Pr_{\ket{\psi} \sim \mu_N}\left[\Alg\left(\ket{\psi}^{\otimes t}\right) = 1\right],
\end{align*}
where the inequality in the second line follows from \Cref{def:state_t_design} and the fact that $A \preceq B$ implies $\tr(MB) - \tr(MA) = \tr(M (B - A)) \ge 0$, because the trace of a product of two positive semidefinite matrices is always nonnegative. This establishes the right inequality in the statement of the lemma; the left inequality follows by following the same steps with the other half of \Cref{def:state_t_design}.
\end{proof}

A similar statement can easily be shown for unitary designs when the only queries made by the algorithm are parallel:

\begin{lemma}
\label{lem:parallel_unitary_t_design_multiplicative}
Let $S$ be an $\eps$-approximate unitary $t$-design over $\Unitaries(N)$, and let $\Alg^U$ be a quantum algorithm whose only queries to $U \in \Unitaries(N)$ consist of a single application of $U^{\otimes t}$. Then:
$$(1 - \eps)\Pr_{U \sim \mu_N}\left[\Alg^U = 1\right] \le \Pr_{U \sim S}\left[\Alg^U = 1\right] \le (1 + \eps) \Pr_{U\sim \mu_N}\left[\Alg^U = 1\right].$$
\end{lemma}

\begin{proof}
Let $\rho$ be the input state of the algorithm, and let $0 \preceq M \preceq \identity$ be the measurement performed at the end of the algorithm, so that
\[
\Pr\left[\Alg^U = 1\right] = \tr\left(M (U^{\otimes t} \otimes I) \rho (U^{\otimes t} \otimes I)^\dagger\right).
\]
Then:
\begin{align*}
\Pr_{U \sim S}\left[\Alg^U = 1\right] &=
\E_{U \sim S} \left[\tr\left(M (U^{\otimes t} \otimes I) \rho (U^{\otimes t} \otimes I)^\dagger\right)\right]\\
&= \tr\left(M \E_{U \sim S}\left[(U^{\otimes t} \otimes I) \rho (U^{\otimes t} \otimes I)^\dagger\right]\right)\\
&\le \tr\left(M (1 + \eps)\E_{U \sim \mu_N}\left[(U^{\otimes t} \otimes I) \rho (U^{\otimes t} \otimes I)^\dagger\right]\right)\\
&= (1 + \eps)\E_{U \sim \mu_N} \left[\tr\left(M (U^{\otimes t} \otimes I) \rho (U^{\otimes t} \otimes I)^\dagger\right)\right]\\
&= (1 + \eps) \Pr_{U \sim \mu_N}\left[\Alg^U = 1\right],
\end{align*}
where the second and fourth lines hold by linearity of expectation, and the inequality in the third line follows from \Cref{def:unitary_t_design}. Specifically, the third line uses the fact that if $A \preceq B$ are superoperators and $\rho$ is positive semidefinite, then $\tr(M \cdot (B \otimes I)(\rho)) - \tr(M \cdot (A \otimes I)(\rho)) = \tr(M \cdot ((B - A) \otimes I)(\rho)) \ge 0$, because $B - A$ is completely positive, and the trace of a product of two positive semidefinite matrices is always nonnegative. This establishes the right inequality in the statement of the lemma; the left inequality follows by following the same steps with the other half of \Cref{def:unitary_t_design}.
\end{proof}

Using an idea from \cite{AMR20-random}, one can straightforwardly generalize \Cref{lem:parallel_unitary_t_design_multiplicative} to algorithms that make \textit{adaptive} queries to $U$, but not controlled-$U$. The key idea is that using quantum gate teleportation, one can simulate adaptive queries to a unitary transformation using parallel queries and postselection. For the next lemma that shows this, recall that the \textit{Choi state} of a unitary $U \in \Unitaries(N)$ is the state
\[
\ket{\phi_U} \coloneqq (U \otimes I)\ket{\Phi_N},
\]
where $\ket{\Phi_N} \coloneqq \frac{1}{\sqrt{N}} \sum_{i=1}^N \ket{i}\ket{i}$ is the standard maximally entangled state.

\begin{lemma}
\label{lem:parallel_simulation_postselection}
Let $\Alg^U$ be a quantum algorithm that makes $t$ adaptive queries to $U \in \Unitaries(N)$ (but not controlled-$U$ or $U^\dagger$). Then there exists an algorithm $\AlgB\left(\ket{\phi_U}^{\otimes t}\right)$ such that:
\[
\Pr\left[\AlgB\left(\ket{\phi_U}^{\otimes t}\right) = 1\right] = \frac{\Pr\left[\Alg^U = 1\right]}{N^{2t}}.
\]
\end{lemma}

\begin{proof}
Consider the following circuit:
\vspace{0.5em}
\[
\Qcircuit @C=1em @R=0.8em @!R {
\lstick{} & \qw \inputgroupv{1}{2}{.7em}{.7em}{\ket{\phi_U}} & \qw & \rstick{U\ket{\psi}} \qw\\
\lstick{} & \multimeasureD{1}{\ket{\Phi_N}}\\
\lstick{\ket{\psi}} & \ghost{\ket{\Phi_N}}\\
}
\]
In words, given an unknown state $\ket{\psi}$ on the bottom register, this circuit initializes $\ket{\phi_U}$ on the top two registers and then postselects on measuring $\ket{\Phi_N}$ on the bottom two registers. Conditioned on postselection succeeding, the output of the top register is exactly $U\ket{\psi}$, as the following calculation shows:
\begin{align*}
(I \otimes \bra{\Phi_N})(\ket{\phi_U} \otimes \ket{\psi})
&=
(U \otimes \bra{\Phi_N})(\ket{\Phi_N} \otimes \ket{\psi})\\
&= \frac{1}{N}\left(U \otimes \sum_{i=1}^N \bra{i}\bra{i} \right) \left(\sum_{i=1}^N \sum_{j=1}^N \psi_j\ket{i}\ket{i}\ket{j} \right)\\
&= \frac{1}{N}\left(U \otimes \sum_{i=1}^N \bra{i}\bra{i} \right) \left(\sum_{i=1}^N \psi_i\ket{i}\ket{i}\ket{i} \right)\\
&= \frac{U}{N}\sum_{i=1}^n \psi_i \ket{i}\\
&= \frac{U\ket{\psi}}{N}.
\end{align*}
This also shows that postselection succeeds with probability $\frac{1}{N^2}$, independent of $U$ or $\ket{\psi}$.

Let $\AlgB$ simulate each query that $\Alg$ makes to $U$ using the above circuit and a copy of $\ket{\phi_U}$. Conditioned on all $t$ postselection steps succeeding, which occurs with probability $\frac{1}{N^{2t}}$, the output state of $\AlgB$ is exactly the same as the output state of $\Alg$. Thus, 
\[
\Pr\left[\AlgB\left(\ket{\phi_U}^{\otimes t}\right) = 1\right] = \frac{\Pr\left[\Alg^U = 1\right]}{N^{2t}}.\qedhere
\]
\end{proof}

An interesting question is whether \Cref{lem:parallel_simulation_postselection} can be generalized to algorithms $\Alg^U$ that can make queries to both $U$ and $U^\dagger$.
In a sense, we are just asking whether queries to $U^\dagger$ can be simulated by a combination of queries to $U$ and postselection. Notably, in the case that $U$ is a real orthogonal matrix, this is possible, because $\ket{\phi_U}$ is equivalent to $\ket{\phi_{U^\top}}$ up to swapping the two registers, and $U^\top = U^\dagger$ for real matrices $U$.

The generalization of \Cref{lem:parallel_unitary_t_design_multiplicative} to algorithms that make adaptive queries now follows.

\begin{lemma}
\label{lem:adaptive_unitary_t_design_multiplicative}
Let $S$ be an $\eps$-approximate unitary $t$-design over $\Unitaries(N)$, and let $\Alg^U$ be a quantum algorithm that makes $t$ queries to $U$ (but not controlled-$U$ or $U^\dagger$). Then:
$$(1 - \eps)\Pr_{U \sim \mu_N}\left[\Alg^U = 1\right] \le \Pr_{U \sim S}\left[\Alg^U = 1\right] \le (1 + \eps) \Pr_{U\sim \mu_N}\left[\Alg^U = 1\right].$$
\end{lemma}

\begin{proof}
Let $\AlgB\left(\ket{\phi_U}^{\otimes t}\right)$ be the algorithm from \Cref{lem:parallel_simulation_postselection}. Then:
\begin{align*}
\Pr_{U \sim S}\left[\Alg^U = 1\right]
&= N^{2t} \Pr_{U \sim S}\left[\AlgB\left(\ket{\phi_U}^{\otimes t}\right) = 1\right]\\
&\le (1 + \eps)N^{2t} \Pr_{U \sim \mu_N}\left[\AlgB\left(\ket{\phi_U}^{\otimes t}\right) = 1\right]\\
&= (1 + \eps)\Pr_{U \sim \mu_N}\left[\Alg^U = 1\right],
\end{align*}
where the first and last lines use \Cref{lem:parallel_simulation_postselection}, and the second line applies \Cref{lem:parallel_unitary_t_design_multiplicative}. The other inequality in the lemma follows by similar steps, using the other half of \Cref{lem:parallel_unitary_t_design_multiplicative}.
\end{proof}

Finally, we wish to extend \Cref{lem:adaptive_unitary_t_design_multiplicative} to algorithms that can also make queries to controlled-$U$. It is not at all obvious that this is possible, because it is well known that queries to controlled-$U$ cannot be simulated efficiently using queries to $U$ \cite{GST20-if}. At the same time, it is unclear how one should pick the phase on controlled-$U$, because the definition of an approximate unitary $t$-design (\Cref{def:unitary_t_design}) ``forgets'' the global phase of $U$. We address both problems simultaneously by observing that the argument using \Cref{lem:parallel_simulation_postselection} can still go through if we choose the phase of $U$ \textit{randomly}. In order for this argument to hold, we assume that our approximate unitary $t$-design $S$ is \textit{phase-invariant}, which we take to mean that for any $U$ sampled from the design, $U$ and $\omega U$ are chosen with the same probability, where $\omega = e^{\frac{2 \pi i}{t+1}}$ is a primitive $(t+1)$th root of unity.\footnote{Actually, our proof only requires that the design be invariant under \textit{any} group of complex units whose first $t$ moments are the same as the uniform measure over complex units. We choose the group generated by $\omega$ only because it is the smallest group with this property.} Note that we can make any unitary design phase-invariant via multiplication by a uniformly random $(t+1)$th root of unity.

\begin{lemma}
\label{lem:controlled_adaptive_unitary_t_design_multiplicative}
Let $S$ be a phase-invariant $\eps$-approximate unitary $t$-design over $\Unitaries(N)$, and let $\Alg^U$ be a quantum algorithm that makes $t$ queries to $U$ (including controlled-$U$, but not $U^\dagger$). Then:
$$(1 - \eps)\Pr_{U \sim \mu_N}\left[\Alg^U = 1\right] \le \Pr_{U \sim S}\left[\Alg^U = 1\right] \le (1 + \eps) \Pr_{U\sim \mu_N}\left[\Alg^U = 1\right].$$
\end{lemma}

\begin{proof}
Let $\langle \omega \rangle$ be the group of $(t+1)$th roots of unity. We first claim that given $t$ copies of the Choi state of $U$ (i.e.\ $\ket{\phi_U}^{\otimes t}$), we can generate the state
\[
\sigma_{U,t} \coloneqq \E_{\varphi \sim \langle \omega \rangle} \left[\ket{\phi_{I \oplus \varphi U}} \bra{\phi_{I \oplus \varphi U}}^{\otimes t} \right],
\]
which is $t$ copies of the Choi state of controlled-$\varphi U$, averaged over all phases $\varphi$ that are $(t+1)$th roots of unity. Assuming this claim holds, then
\begin{align*}
\Pr_{U \sim S}\left[\Alg^U = 1\right]
&= \E_{U \sim S} \left[ \Pr_{\varphi \sim \langle \omega \rangle}\left[\Alg^{\varphi U} = 1\right] \right]\\
&= \E_{U \sim S}\left[(2N)^{2t} \Pr_{\varphi \sim \langle \omega \rangle}\left[\AlgB\left(\ket{\phi_{I \oplus \varphi U}}^{\otimes t}\right) = 1\right]\right]\\
&= (2N)^{2t} \Pr_{U \sim S}\left[\AlgB\left(\sigma_{U,t} \right) = 1\right]\\
&\le (1 + \eps)(2N)^{2t} \Pr_{U \sim \mu_N}\left[\AlgB\left(\sigma_{U,t} \right) = 1\right]\\
&= (1 + \eps)\Pr_{U \sim \mu_N}\left[\Alg^U = 1\right],
\end{align*}
where the first line applies the phase-invariance of the design, the second line holds for the algorithm $\AlgB$ defined in \Cref{lem:parallel_simulation_postselection}, the third line holds by the definition of $\sigma_{U, t}$, the fourth line applies \Cref{lem:parallel_unitary_t_design_multiplicative} and the claim that $\sigma_{U, t}$ can be prepared from $\ket{\phi_U}^{\otimes t}$, and the last line appeals to \Cref{lem:parallel_simulation_postselection} and the phase-invariance of the Haar measure. The other inequality in the statement of the lemma follows from similar steps, so for the rest of the proof we turn to proving the claim. We remark that the remainder of the proof is closely related to the proof of \cite[Lemma 15]{Kre21-tsirelson}, which also involves showing how to prepare a state obtained by averaging over phases.

Recall that $\ket{\phi_{I \oplus \varphi U}}$ is defined by 
\[
\ket{\phi_{I \oplus \varphi U}} \coloneqq \frac{1}{\sqrt{2N}} \left(\sum_{i=1}^N \ket{i}\ket{i} + \sum_{i=N}^{2N} \varphi U\ket{i}\ket{i}\right).
\]
By identifying $[2N]$ with $\{0,1\} \times [N]$, we can also write this in the form
\[
\ket{\phi_{I \oplus \varphi U}} \equiv \frac{1}{\sqrt{2N}} \left(\sum_{i=1}^N \ket{0}\ket{i}\ket{0}\ket{i} + \sum_{i=1}^{N} \varphi U\ket{1}\ket{i}\ket{1}\ket{i}\right).
\]
Swapping the ordering of the second and third registers, we identify this state with
\[
\ket{\phi_{I \oplus \varphi U}} \equiv \frac{\ket{00}\ket{\Phi_N} + \varphi\ket{11}\ket{\phi_U}}{\sqrt{2}}.
\]
For convenience, define $\ket{\psi_0} \coloneqq \ket{00}\ket{\Phi_N}$ and $\ket{\psi_1} \coloneqq \ket{11}\ket{\phi_U}$. For $x \in \{0,1\}^t$, we extend this notation via $\ket{\psi_x} \coloneqq \bigotimes_{i=1}^t \ket{\psi_{x_i}}$. For $i = 0,\ldots,t$, define
\[
\ket{\overline{\psi}_i} \coloneqq \binom{t}{i}^{-1/2} \sum_{\substack{ x \in \{0,1\}^t: \\ \wt(x) = i}} \ket{\psi_x}
\]
We claim that $\sigma_{U,t} = \rho_{U,t}$, where
\[
\rho_{U,t} \coloneqq \sum_{i = 1}^t \frac{\binom{t}{i}}{2^t} \ket{\overline{\psi}_i}\bra{\overline{\psi}_i}.
\]
To see why, note that both $\sigma_{U,t}$ and $\rho_{U,t}$ are only supported on the orthonormal basis $\{\ket{\psi_x} : x \in \{0,1\}^t\}$, so it suffices to show $\bra{\psi_x}\sigma_{U,t}\ket{\psi_y} = \bra{\psi_x}\rho_{U,t}\ket{\psi_y}$ for each $x, y \in \{0,1\}^t$. Observe that
\begin{align*}
\bra{\psi_x}\sigma_{U,t}\ket{\psi_y}
&= \E_{\varphi \sim \langle \omega \rangle} \left[ \left(\prod_{i=1}^t \frac{1}{\sqrt{2}}\varphi^{x_i}\right) \left(\prod_{i=1}^t \frac{1}{\sqrt{2}}\varphi^{-y_i}\right) \right]\\
&= \frac{1}{2^t} \cdot \E_{\varphi \sim \langle \omega \rangle} \left[\varphi^{\wt(x) - \wt(y)} \right]\\
&= \begin{cases}
\frac{1}{2^t} & \wt(x) = \wt(y)\\
0 & \wt(x) \neq \wt(y).
\end{cases}
\end{align*}
Above, we are using the fact that the first $t$ moments of $\langle \omega \rangle$ are the same as the first $t$ moments of the full group of complex units.

Clearly, $\bra{\psi_x}\rho_{U,t}\ket{\psi_y} = \bra{\psi_x}\sigma_{U,t}\ket{\psi_y} = 0$ whenever $\wt(x) \neq \wt(y)$, because $\rho_{U,t}$ is a mixture of pure states that are each a superposition over basis states that have the same Hamming weight. On the other hand, when $\wt(x) = \wt(y) = i$, we have
\[
\bra{\psi_x}\rho_{U,t}\ket{\psi_y}
= \frac{\binom{t}{i}}{2^n} \braket{\psi_x|\overline{\psi}_i}\braket{\overline{\psi}_i|\psi_y} = \frac{1}{2^t},
\]
as claimed.

To complete the proof, we only have to show how to produce the state $\rho_{U,t}$ from $t$ copies of $\ket{\phi_U}$. Since $\rho_{U,t}$ is a probabilistic mixture of the $\ket{\overline{\psi}_i}$s, it suffices to show how to produce $\ket{\overline{\psi}_i}$. Begin by initializing the state
\[
\binom{t}{i}^{-1/2} \sum_{\substack{ x \in \{0,1\}^t: \\ \wt(x) = i}} \ket{x}\ket{\psi_{1^i 0^{t-i}}},
\]
which can be viewed as a tensor product of $i$ copies of $\ket{\phi_U}$ and a fixed state independent of $\ket{\phi_U}$. Using permutations on the second register controlled on the first register, the above state can be mapped to
\[
\binom{t}{i}^{-1/2} \sum_{\substack{ x \in \{0,1\}^t: \\ \wt(x) = i}} \ket{x}\ket{\psi_{x}}.
\]
Finally, we can erase the $\ket{x}$ by, for each $i$, flipping the $i$th bit of the first register controlled on the $i$th part of the second register being orthogonal to $\ket{\psi_0}$. This works because $\ket{\psi_0}$ is a known state and $\ket{\psi_1}$ is orthogonal to $\ket{\psi_0}$. Doing so leaves us with:
\[
\binom{t}{i}^{-1/2} \sum_{\substack{ x \in \{0,1\}^t: \\ \wt(x) = i}} \ket{0^t}\ket{\psi_{x}},
\]
which is just $\ket{0^t}\ket{\overline{\psi}_i}$.
\end{proof}

\section{Breaking pseudorandomness with a classical oracle}

In this section, we prove that a polynomial-time quantum algorithm with a $\mathsf{PP}$ oracle can distinguish a pseudorandom state from a Haar-random state. First, we need a lemma about the overlap between a fixed state $\ket{\varphi}$ and a Haar-random state $\ket{\psi}$.

\begin{lemma}
\label{lem:haar_overlap_small}
Let $\ket{\varphi} \in \States(N)$, and let $\eps > 0$. Then:\footnote{As observed in \cite{CGGH+23-one-way}, an earlier version of this work \cite{conf-version} gave the incorrect bound $e^{-\eps N}$ here, which was carried over from \cite[Equation (14)]{BHH16-designs}.}
$$\Pr_{\ket{\psi} \sim \sigma_{N}} \left[ |\braket{\psi|\varphi}|^2 \ge \eps \right] \le e^{-\eps (N-1)} .$$
\end{lemma}

\begin{proof}
It is well known that if a Haar-random state $\ket{\psi}$ is measured in any fixed basis (say, a basis containing $\ket{\varphi}$), the measurement probabilities are uniform over the $N$-dimensional probability simplex, or equivalently sampled according to a standard Dirichlet distribution. $|\braket{\psi|\varphi}|^2$ is one of the marginals of this Dirichlet distribution, and hence it is distributed as $|\braket{\psi|\varphi}|^2 \sim \mathrm{Beta}(1, N - 1)$. The probability density function of this distribution is given by
\[
p(x) = (N - 1)(1 - x)^{N - 2}.
\]
It follows that
\begin{align*}
\Pr_{\ket{\psi} \sim \sigma_{N}} \left[ |\braket{\psi|\varphi}|^2 \ge \eps \right] &= \int_\eps^1 (N - 1)(1 - x)^{N - 2} \mathrm{d}x\\
&= (1 - \eps)^{N - 1}\\
&\le e^{-\eps (N - 1)}.\qedhere
\end{align*}
\end{proof}

The formal statement of our result is below.

\begin{theorem}
\label{thm:pp_oracle}
For any PRS ensemble $\{\ket{\varphi_k}\}_{k \in \Key}$ of $n$-qubit states with security parameter $\kappa$ satisfying $n = \omega(\log \kappa)$, there exists a $\mathsf{PP}$ language $\Lang$, a $\poly(\kappa)$-time quantum algorithm $\Alg^\Lang$, and $T = \poly(\kappa)$ such that the following holds. Let $X \sim \{0,1\}$ be a uniform random bit. Let $\ket{\psi}$ be sampled uniformly from the PRS ensemble if $X = 0$, and otherwise let $\ket{\psi}$ be sampled from the Haar measure $\sigma_{2^n}$ if $X = 1$. Then we have:
$$\Pr_{X,\ket{\psi}}\left[\Alg^{\Lang}\left(1^\kappa, \ket{\psi}^{\otimes T}\right) = X\right] \ge 0.995.$$
\end{theorem}

\begin{proof}
We first describe $\Alg$. For some $T$ to be chosen later, on input $\ket{\psi}^{\otimes T}$, $\Alg$ measures each copy of $\ket{\psi}$ in a different randomly chosen Clifford basis. Call the list of measurement bases $b = (b_1, b_2,\ldots, b_T)$ and the measurement results $c = (c_1, c_2, \ldots, c_T)$. $\Alg$ then feeds $(b,c)$ into a single query to $\Lang$, and outputs the result of the query. This takes polynomial time because there exists an $O(n^3)$-time algorithm to sample a random $n$-qubit Clifford unitary, and this algorithm also produces an implementation of the unitary with $O(n^2 / \log n)$ gates \cite{KS14-random-clifford,AG04-stabilizer}. 

The $\mathsf{PP}$ language $\Lang$ that we choose is most easily described in terms of a $\mathsf{PromisePostBQP}$ algorithm $\mathcal{B}(b, c)$ (i.e.\ a postselected polynomial-time quantum algorithm, as in \Cref{def:PostBQP}), by the equivalence $\mathsf{PromisePostBQP} = \mathsf{PromisePP}$ (\Cref{lem:postbqp-pp}). That is, we specify an algorithm $\mathcal{B}(b, c)$ that outputs a trit in $\{0, 1, *\}$, and this algorithm defines a promise problem $\Pi \in \mathsf{PromisePostBQP}$ as follows:
\begin{enumerate}[(i)]
    \item If $\Pr[\mathcal{B}(b, c) \in \{0, 1\}] > 0$ and $\Pr[\mathcal{B}(b, c) = 1 \mid \mathcal{B}(b, c) \in \{0, 1\}] \ge \frac{2}{3}$, then $\Pi(b, c) = 1$.
    \item If $\Pr[\mathcal{B}(b, c) \in \{0, 1\}] > 0$ and $\Pr[\mathcal{B}(b, c) = 1 \mid \mathcal{B}(b, c) \in \{0, 1\}] \le \frac{1}{3}$, then $\Pi(b, c) = 0$.
    \item Otherwise, $\Pi(b, c) = \bot$.
\end{enumerate}
By Aaronson's theorem (\Cref{lem:postbqp-pp}), $\Pi$ is also in $\mathsf{PromisePP}$. Because $\mathsf{PP}$ is a syntactic class, the promise problem $\Pi$ can be extended to a language $\Lang \in \mathsf{PP}$.

Let $S$ be a $\frac{1}{17}$-approximate $n$-qubit quantum $T$-design (\Cref{def:state_t_design}) such that a state can be drawn from $S$ in $\poly(\kappa)$ time (because $n,T \le \poly(\kappa)$, the existence of such a design follows from \Cref{prop:unitary_to_state_design} and \Cref{lem:efficient_unitary_designs}). $\mathcal{B}$ begins by initializing the state:
$$\hat{\rho} \coloneqq \frac{1}{2}\ket{0}\bra{0} \otimes \E_{k \sim \Key}\left[\ket{\varphi_k}\bra{\varphi_k}^{\otimes T}\right] + \frac{1}{2}\ket{1}\bra{1} \otimes \E_{\ket{\phi} \sim S}\left[\ket{\phi}\bra{\phi}^{\otimes T}\right].$$
$\mathcal{B}$ measures all but the leftmost qubit of $\hat{\rho}$ in the basis given by $b$, and postselects on observing $c$ (i.e.\ $\mathcal{B}$ outputs $*$ if the measurements are not equal to $c$). Finally, conditioned on postselection succeeding, $\mathcal{B}$ measures and outputs the result of the leftmost qubit that was not measured.

It remains to show that $\Alg$ distinguishes the pseudorandom and Haar-random state ensembles. For the purpose of this analysis, it will be convenient to view $\hat{\rho}$ as an approximation to the state:
$$\rho \coloneqq \frac{1}{2}\ket{0}\bra{0} \otimes \E_{k \sim \Key}\left[\ket{\varphi_k}\bra{\varphi_k}^{\otimes T}\right] + \frac{1}{2}\ket{1}\bra{1} \otimes \E_{\ket{\phi} \sim \sigma_{2^n}}\left[\ket{\phi}\bra{\phi}^{\otimes T}\right],$$
where the $\eps$-approximate $T$-design $S$ is replaced by the Haar measure $\sigma_{2^n}$. Indeed, we will essentially argue the algorithm's correctness if the state $\hat{\rho}$ is replaced by $\rho$, and then argue that this implies the correctness of the actual algorithm.

For each $k \in \Key$, define $O_k \coloneqq \ket{\varphi_k}\bra{\varphi_k}$. Note that if $X = 0$ (i.e.\ $\ket{\psi}$ is pseudorandom), there always exists a $k$ such that $\tr(O_k \ket{\psi}\bra{\psi}) = 1$, namely whichever $k$ satisfies $\ket{\psi} = \ket{\varphi_k}$. On the other hand, by \Cref{lem:haar_overlap_small} and a union bound, if $X = 1$ (i.e.\ $\ket{\psi}$ is Haar-random), $\tr(O_k \ket{\psi}\bra{\psi}) < \frac{1}{3}$ for every $k \in \Key$, except with probability at most $2^\kappa \cdot e^{-(2^n - 1)/3}$ over $\ket{\psi}$. This probability is negligible in $\kappa$ because $n = \omega(\log \kappa)$, by assumption.

If we choose $M = |\Key| = 2^\kappa$, $\eps = \frac{1}{3}$, and $\delta = 0.001 - 2^k \cdot e^{-(2^n - 1)/3}$, then by \Cref{thm:HKP} there exists a quantum algorithm that takes as input the results $(b, c)$ of $T = O(\kappa)$ single-copy random Clifford measurements of $\ket{\psi}$, uses the measurement results to estimate $\tr(O_k \ket{\psi}\bra{\psi})$ for each $k$ up to additive error $\frac{1}{3}$, and is correct with probability at least $0.999 + 2^\kappa \cdot e^{-(2^n - 1)/3}$. In particular, this algorithm can distinguish the pseudorandom ensemble from the Haar-random ensemble, by checking if there exists a $k$ such that the estimate for $\tr(O_k \ket{\psi}\bra{\psi})$ is at least $\frac{2}{3}$. Call this algorithm $\mathcal{C}$, so that $\Pr[\mathcal{C}(b, c) = X] \ge 0.999$.

We will not actually use $\mathcal{C}$, but only its existence. By the optimality of the Bayes decision rule (\Cref{lem:bayes_decision}), because $\mathcal{C}$ uses $(b, c)$ to identify a state $\ket{\psi}$ as either Haar-random or pseudorandom with probability $0.999$, an algorithm that computes the maximum a posteriori estimate of $X$ also succeeds with probability at least $0.999$. In symbols, let $p_i = \Pr[X = i \mid b, c]$, which we view as a random variable (depending on $b$ and $c$) for each $i \in \{0,1\}$. Then $\Pr\left[\argmax_i p_i = X\right] \ge 0.999$.

Next, observe that $\Pr\left[\argmax_i p_i = X\right] = \E\left[\Pr\left[\argmax_i p_i = X | b, c\right]\right] = \E\left[\max_i p_i\right]$, by the law of total expectation. So, by Markov's inequality (and the fact that $\max_i p_i \le 1$), we know $\Pr\left[\max_i p_i \ge \frac{3}{4}\right] \ge 0.996$. In other words, the Bayes decision rule is usually at least $75\%$ confident in its predictions, so to speak.

Notice that $p_i$ equals the probability (conditioned on postselection succeeding) that $\mathcal{B}$ outputs $i$ if it starts with $\rho$ in place of $\hat{\rho}$. For $i \in \{0,1\}$, define $\hat{p}_i$ analogously as the postselcted output probabilities of $\mathcal{B}$ itself: $\hat{p}_i \coloneqq \Pr\left[\mathcal{B}(b, c) = i \mid \mathcal{B}(b, c) \in \{0,1\}\right]$. To argue that $\Alg$ is correct with $0.995$ probability, it suffices to show that 
\[
\Pr\left[\max_i \hat{p}_i \ge \frac{2}{3} \land \argmax_i \hat{p}_i = X\right] \ge 0.995,\]
as in this case the $\mathsf{PromisePostBQP}$ promise is satisfied and the output of $\Lang$ agrees with $X$. We have that:
\begin{align*}
\Pr\left[\max_i \hat{p}_i \ge \frac{2}{3} \land \argmax_i \hat{p}_i = X\right] &\ge \Pr\left[\max_i p_i \ge \frac{3}{4} \land \argmax_i p_i = X\right]\\
&\ge 1 - \Pr\left[\max_i p_i < \frac{3}{4}\right] - \Pr\left[\argmax_i p_i \neq X\right]\\
&\ge 0.996 - \Pr\left[\argmax_i p_i \neq X\right]\\
&\ge 0.995.
\end{align*}
Above, the first inequality follows from the assumption that $S$ is a $\frac{1}{17}$-approximate $T$-design, because the acceptance probability of a postselected quantum algorithm can be viewed as the ratio of two probabilities:
$$\hat{p_i} = \frac{\Pr[\mathcal{B}(b, c) = i]}{\Pr[\mathcal{B}(b, c) \in \{0,1\}]}.$$
\Cref{lem:state_t_design_multiplicative} implies that both the numerator and denominator change by at most a multiplicative factor of $1 \pm \frac{1}{17}$ when switching between $\rho$ and $\hat{\rho}$. So, if $p_i \ge \frac{3}{4}$, then $\hat{p_i} \ge \frac{3}{4} \cdot \frac{1 - \frac{1}{17}}{1 + \frac{1}{17}} = \frac{2}{3}$. The second inequality follows by a union bound, and the remaining inequalities were established above.
\end{proof}

We remark that the above theorem also holds relative to all oracles, in the sense that if the state generation algorithm $G$ in the definition of the PRS (\Cref{def:prs}) queries a classical or quantum oracle $\UOracle$, then the corresponding ensemble of states can be distinguished from Haar-random by a polynomial-time quantum algorithm with a $\mathsf{PromisePostBQP}^\UOracle$ oracle.

\section{Pseudorandomness from a quantum oracle}

In this section, we construct a quantum oracle $(\UOracle, \COracle)$ relative to which $\mathsf{PromiseBQP} = \mathsf{PromiseQMA}$ and PRUs exist. Let us first describe the oracle, which consists of two parts: a quantum oracle $\UOracle$ and a classical oracle (i.e.\ a language) $\COracle$.

\subsection{Definition of the oracle}

The oracle $\UOracle$ is a sequence of unitaries $\{\UOracle_n\}_{n \in \Naturals}$, where each $\UOracle_n$ is a direct sum of $2^n$ different Haar-random $n$-qubit unitaries. In other words, for each $n$ we sample $\UOracle_n \sim \mu_{2^n}^{2^n}$. We denote this distribution over $\UOracle$ by $\UOracle \sim \mathcal{D}$.

We construct the language $\COracle$ deterministically and independently of $\UOracle$. We specify the language in stages: first we define $\COracle$'s behavior on the $1$-bit strings, then the $2$-bit strings, then the $3$-bit strings, and so on. For a string $x$, we define $\COracle(x) = 1$ if the following all hold:
\begin{enumerate}[(1)]
    \item $x$ is a description of a quantum oracle circuit $\Ver^{\overline{\UOracle}, \COracle}(\ket{\psi})$ that takes a quantum state $\ket{\psi}$ as input, and makes queries to a quantum oracle $\overline{\UOracle}$ and the classical oracle $\COracle$. Note that $\ket{\psi}$ and $\overline{\UOracle}$ are not part of the description of $\Ver$; they are auxiliary inputs.
    \item $\mathcal{V}$ runs in time at most $|x| - 1$, and hence can query $\COracle$ on inputs of length at most $|x| - 1$.
    \item The average acceptance probability of $\Ver$ (viewed as a $\mathsf{QMA}$ verifier) is greater than $1/2$ when averaged over $\overline{\UOracle} \sim \mathcal{D}$. In symbols, we mean precisely:
    \[\E_{\overline{\UOracle} \sim \mathcal{D}}\left[ \max_{\ket{\psi}} \Pr[\Ver^{\overline{\UOracle}, \COracle}(\ket{\psi}) = 1 ] \right] > \frac{1}{2}.\]
\end{enumerate}

Condition (2) guarantees that $\COracle$ is not circularly defined, because the quantity in condition (3) depends only on the previously constructed parts of the oracle. Notice also the care we have used in our notation: $\overline{\UOracle}$ is merely used to take an average in the definition of $\COracle$; it is not the same as $\UOracle$.

\subsection{\texorpdfstring{$\mathsf{PromiseBQP} = \mathsf{PromiseQMA}$}{PromiseBQP = PromiseQMA} relative to \texorpdfstring{$(\UOracle,\COracle)$}{(U, C)}}
Now we turn to showing that our oracle satisfies the desired properties.
We start with a lemma showing that the acceptance probability of a quantum query algorithm, viewed as a function of the unitary transformation used in the query, is Lipschitz.

\begin{lemma}\label{lem:query_lipschitz}
Let $\Alg^U$ be a quantum algorithm that makes $T$ queries to $U \in \Unitaries(D)$. Define $f: \Unitaries(D) \to \Reals$ by $f(U) \coloneqq \Pr\left[\Alg^U = 1 \right]$. Then $f$ is $T$-Lipschitz in the Frobenius norm.
\end{lemma}

\begin{proof}
Suppose that $||U - V||_F \le d$. Then $||I \oplus U - I \oplus V||_F \le d$, and also $||I \oplus U^\dagger - I \oplus V^\dagger||_F \le d$, recalling that $I \oplus U$ is controlled-$U$ and $I \oplus V$ is controlled-$V$.
By \Cref{lem:frobenius_diamond}, this implies that the distance between controlled-$U$ and controlled-$V$ in the diamond norm is at most $2d$ (and likewise for controlled-$U^\dagger$ and controlled-$V^\dagger$). The sub-additivity of the diamond norm under composition implies that as superoperators, $||\Alg^U - \Alg^V||_\diamond \le 2Td$. By \Cref{fact:diamond_bounds_trace}, we conclude that $|f(U) - f(V)| \le Td$.
\end{proof}

The next lemma extends \Cref{lem:query_lipschitz} to $\mathsf{QMA}$ verifiers: we should think of $\Ver$ as a $\mathsf{QMA}$ verifier that receives a witness $\ket{\psi}$, in which case this lemma states that the maximum acceptance probability of $\Ver$ is Lipschitz with respect to the queried unitary.

\begin{lemma}
\label{lem:qma_query_lipschitz}
Let $\Ver^U(\ket{\psi})$ be a quantum algorithm that makes $T$ queries to $U \in \Unitaries(D)$ and takes as input a quantum state $\ket{\psi}$ on some fixed (but arbitrary) number of qubits. Define $f: \Unitaries(D) \to \Reals$ by $f(U) \coloneqq \max_{\ket{\psi}}\Pr\left[\Ver^U(\ket{\psi}) = 1 \right]$. Then $f$ is $T$-Lipschitz in the Frobenius norm.
\end{lemma}

\begin{proof}
Note that $f$ is well-defined because of the extreme value theorem. Define $f_\psi: \Unitaries(D) \to \Reals$ by:
$$f_\psi(U) \coloneqq \Pr\left[\Ver^U(\ket{\psi}) = 1 \right],$$
so that $f(U) = \max_{\ket{\psi}} f_\psi(U)$.
\Cref{lem:query_lipschitz} implies that $f_\psi$ is $T$-Lipschitz for every $\ket{\psi}$. Let $U, V \in \Unitaries(D)$, and suppose that $\ket{\psi}$ and $\ket{\varphi}$ are such that $f(U) = f_\psi(U)$ and $f(V) = f_\varphi(V)$. Then:
\begin{align*}
|f(U) - f(V)| &= |f_\psi(U) - f_\varphi(V)|\\
&= \max\{f_\psi(U) - f_\varphi(V), f_\varphi(V) - f_\psi(U)\}\\
&\le \max\{f_\psi(U) - f_\psi(V), f_\varphi(V) - f_\varphi(U)\}\\
&\le T||U - V||_F,
\end{align*}
where the third line uses the fact that $f_\psi(V) \le f_\varphi(V)$ and $f_\varphi(U) \le f_\psi(U)$, and the last line uses the fact that $f_\psi$ and $f_\varphi$ are $T$-Lipschitz.
\end{proof}

We are ready to prove the first main result of this section, that $\mathsf{PromiseBQP}^{\UOracle,\COracle} = \mathsf{PromiseQMA}^{\UOracle,\COracle}$.

\begin{theorem}
\label{thm:bqp^u=qma^u}
With probability $1$ over $\UOracle \sim \mathcal{D}$, $\mathsf{PromiseBQP}^{\UOracle, \COracle} = \mathsf{PromiseQMA}^{\UOracle, \COracle}$.
\end{theorem}

\begin{proof}
Let $\Pi \in \mathsf{PromiseQMA}^{\UOracle, \COracle}$, which means there exists a polynomial-time verifier $\Ver^{\UOracle, \COracle}(x, \ket{\psi})$ with completeness $\frac{2}{3}$ and soundness $\frac{1}{3}$ according to \Cref{def:QMA}. Without loss of generality, we can amplify the completeness and soundness probabilities of $\Ver$ to $\frac{11}{12}$ and $\frac{1}{12}$, respectively. Let $p(n)$ be a polynomial upper bound on the running time of $\Ver$ on inputs $x$ of length $n$.

We now describe a $\mathsf{PromiseBQP}^{\UOracle, \COracle}$ algorithm $\Alg^{\UOracle, \COracle}(x)$ such that, with probability $1$ over $\UOracle$, $\Alg$ computes $\Pi$ on all but finitely many inputs $x \in \Dom(\Pi)$. The steps of $\Alg$ are:

\begin{enumerate}[(1)]
    \item Let $d \coloneqq \lfloor\log_2 \left(3456|x| p(|x|)^2 + 2\right)\rfloor$. For each $n \in [d]$, $\Alg$ performs process tomography on each $\UOracle_n$, producing estimates $\widetilde{\UOracle}_n$ such that $||\widetilde{\UOracle}_n \cdot \widetilde{\UOracle}_n^\dagger - \UOracle_n \cdot \UOracle_n^\dagger||_\diamond \le \frac{1}{6p(|x|)}$ for every $n$,with probability at least $\frac{2}{3}$ over the randomness of $\Alg$.\footnote{Specifically, one can use the algorithm of \cite{HKOT23-diamond} to estimate each $\UOracle_n$ to $2^{-\Omega(n)}$ error in diamond norm distance. The estimated unitary transformation $\widetilde{\UOracle}_n$ can then be compiled to a circuit using $2^{O(n)}$ $1$- and $2$-qubit gates \cite{VMS04-gates}. Since $n \le d = O(\log |x|)$, this can be done in polynomial time.
    
    Note also that we are using shorthand here: we should really perform process tomography on controlled-$\UOracle_n$, so that $||I \oplus \widetilde{\UOracle}_n \cdot I \oplus \widetilde{\UOracle}_n^\dagger - I \oplus \UOracle_n \cdot I \oplus \UOracle_n^\dagger||_\diamond \le \frac{1}{6p(|x|)}$. We use this same shorthand further below.} We denote the collection of estimates by $\widetilde{\UOracle} \coloneqq \{\widetilde{\UOracle}_n\}_{n \in [d]}$
    \item Next, $\Alg$ constructs a description $x$ of a quantum oracle circuit $\VerW^{\overline{\UOracle}, \COracle}(x, \widetilde{\UOracle}; \ket{\psi})$. This $\VerW$ has $x$ and the unitaries in $\widetilde{\UOracle}$ hard-coded into its description, takes an auxiliary input $\ket{\psi}$,\footnote{This distinction is why the last argument $\ket{\psi}$ is separated with a semicolon.} and queries oracles $\overline{\UOracle}$ and $\COracle$.
    On input $\ket{\psi}$, $\VerW^{\overline{\UOracle}, \COracle}(x, \widetilde{\UOracle}; \ket{\psi})$ replicates the behavior of $\Ver^{\UOracle, \COracle}(x, \ket{\psi})$, except that for each $n \in [d]$, queries to $\UOracle_n$ are replaced by $\widetilde{\UOracle}_n$, and for each $n \in [p(|x|)] \setminus [d]$, queries to $\UOracle_n$ are replaced by queries to $\overline{\UOracle}_n$.
    \item Finally, $\Alg$ queries $\COracle(x)$ and outputs the result.
\end{enumerate}

We now show that for any $x \in \Dom(\Pi)$, with high probability over $\UOracle$, $\Alg$ correctly decides $\Pi$ on $x$, which is to say that $\Pr\left[\Alg^{\UOracle,\COracle}(x) = \Pi(x)\right] \ge \frac{2}{3}$.

For a fixed $x$, given sequences of unitaries $\widetilde{\UOracle} = \{\widetilde{\UOracle}_n\}_{n \in [d]}$ and $\overline{\UOracle} = \{\overline{\UOracle}_n\}_{n \in [p(|x|)] \setminus [d]}$,
define
\[
f(\widetilde{\UOracle}, \overline{\UOracle}) \coloneqq \max_{\ket{\psi}} \Pr\left[\VerW^{\overline{\UOracle}, \COracle}(x, \widetilde{\UOracle}; \ket{\psi}) = 1 \right].
\]
Note that, in this notation, $\Alg$ outputs $1$ if and only if
\begin{equation}
\label{eq:avg_w_half}
\E_{\overline{\UOracle} \sim \mathcal{D}}\left[f(\widetilde{\UOracle}, \overline{\UOracle})\right] > \frac{1}{2}.
\end{equation}
By contrast, the $\mathsf{QMA}$ acceptance probability of $\mathcal{V}$ itself may be written consistently with this notation as:
\begin{equation}
\label{eq:acc_v}
f(\UOracle, \UOracle) = \max_{\ket{\psi}}\Pr\left[\Ver^{\UOracle, \COracle}(x, \ket{\psi}) = 1\right].
\end{equation}
In effect, our goal is to show that \Cref{eq:avg_w_half} gives a good estimator for \Cref{eq:acc_v}. We will do so in two steps: we first show that replacing $\UOracle$ in $f$'s second argument with an average over $\overline{\UOracle}$ approximately preserves the $\mathsf{QMA}$ acceptance probability, and then we argue similarly when replacing $\UOracle$ by the estimate $\widetilde{\UOracle}$ in $f$'s first argument.

By \Cref{lem:qma_query_lipschitz}, $f$ is $p(|x|)$-Lipschitz with respect to the second argument $\overline{\UOracle}$, viewed as a direct sum of matrices $\overline{\UOracle} \equiv \bigoplus_{n=d+1}^{p(|x|)} \overline{\UOracle}_n$.\footnote{This is because each query to a single $\overline{\UOracle}_n$ may be simulated via one query to the entire direct sum.} Hence, from \Cref{thm:haar_concentration} with $N = 3456|x|p(|x|)^2 + 2$, $L = p(|x|)$, and $t = \frac{1}{12}$, we have that:
\begin{align}
    \Pr_{\UOracle\sim \mathcal{D}}\left[\left|f(\UOracle, \UOracle) - \E_{\overline{\UOracle} \sim \mathcal{D}}[f(\UOracle, \overline{\UOracle})]\right| \ge \frac{1}{12}\right]
    &\le 2\exp\left(-\frac{(N-2)t^2}{24L^2}\right)\nonumber\\
    &= 2\exp\left(-\frac{3456|x|p(|x|)^2 \cdot \frac{1}{144}}{24p(|x|)^2}\right)\nonumber\\
    &= 2e^{-|x|}.\label{eq:haar_conc_in_oracle}
\end{align}
The factor of $2$ appears because \Cref{thm:haar_concentration} applies to one-sided error, but the absolute value forces us to consider two-sided error.

Because $\VerW$ calls $\widetilde{\UOracle}$ at most $p(|x|)$ times, and because diamond distance between unitary channels is preserved under taking inverses, 
\Cref{fact:diamond_bounds_trace} implies that for any $\ket{\psi}$,
\[
\abs{\Pr\left[\VerW^{\overline{\UOracle}, \COracle}(x, \widetilde{\UOracle}; \ket{\psi}) = 1 \right] - \Pr\left[\VerW^{\overline{\UOracle}, \COracle}(x, \UOracle; \ket{\psi}) = 1 \right]} \le \frac{p(|x|)}{2} ||\widetilde{\UOracle}_n \cdot \widetilde{\UOracle}_n^\dagger - \UOracle_n \cdot \UOracle_n^\dagger||_\diamond.
\]
Hence, we also have
\begin{align*}
\abs{f(\widetilde{\UOracle}, \overline{\UOracle}) - f(\UOracle, \overline{\UOracle})}
&= 
\abs{\max_{\ket{\psi}}\Pr\left[\VerW^{\overline{\UOracle}, \COracle}(x, \widetilde{\UOracle}; \ket{\psi}) = 1 \right] - \max_{\ket{\psi}}\Pr\left[\VerW^{\overline{\UOracle}, \COracle}(x, \UOracle; \ket{\psi}) = 1 \right]}\\
&\le \frac{p(|x|)}{2} ||\widetilde{\UOracle}_n \cdot \widetilde{\UOracle}_n^\dagger - \UOracle_n \cdot \UOracle_n^\dagger||_\diamond,
\end{align*}
and therefore, by Jensen's inequality,
\[
\abs{\E_{\overline{U} \sim \mathcal{D}}[f(\widetilde{\UOracle}, \overline{\UOracle})] - \E_{\overline{U} \sim \mathcal{D}}[f(\UOracle, \overline{\UOracle})]} \le \frac{p(|x|)}{2} ||\widetilde{\UOracle}_n \cdot \widetilde{\UOracle}_n^\dagger - \UOracle_n \cdot \UOracle_n^\dagger||_\diamond.
\]
Because the estimates $\widetilde{\UOracle_n}$ satisfy $||\widetilde{\UOracle}_n \cdot \widetilde{\UOracle}_n^\dagger - \UOracle_n \cdot \UOracle_n^\dagger||_\diamond \le \frac{1}{6p(|x|)}$ with probability at least $\frac{2}{3}$ over the randomness of $\Alg$, we see:
\[
\Pr_{\Alg}\left[\left|\E_{\overline{\UOracle} \sim \mathcal{D}}\left[f(\UOracle, \overline{\UOracle})\right] - \E_{\overline{\UOracle} \sim \mathcal{D}}\left[f(\widetilde{\UOracle}, \overline{\UOracle})\right] \right| \ge \frac{1}{12} \right] \le \frac{1}{3}.
\]
Combining with \Cref{eq:haar_conc_in_oracle}, and recalling the acceptance criterion of $\Alg$ from \Cref{eq:avg_w_half}, we conclude that except with probability at most $2e^{-|x|}$ over $\UOracle$,
\[
\Pr_{\Alg}\left[\left|f(\UOracle, \UOracle) - \E_{\overline{\UOracle} \sim \mathcal{D}}\left[f(\widetilde{\UOracle}, \overline{\UOracle})\right] \right| \ge \frac{1}{6} \right] \le \frac{1}{3}.
\]
So, except with probability $2e^{-|x|}$ over $\UOracle$:
\begin{align*}
\Pi(x) = 1 \quad&\implies\quad f(\UOracle, \UOracle) \ge \frac{11}{12} \quad\implies\quad \Pr\left[\Alg^{\UOracle,\COracle}(x) = 1\right] \ge \frac{2}{3}\\
\Pi(x) = 0 \quad&\implies\quad f(\UOracle, \UOracle) \le \frac{1}{12} \quad\implies\quad \Pr\left[\Alg^{\UOracle,\COracle}(x) = 0\right] \ge \frac{2}{3}.
\end{align*}
This is to say that $\Alg$ correctly decides $\Pi(x)$, expect with probability at most $2e^{-|x|}$ over $\UOracle$. By the Borel--Cantelli lemma (\Cref{lem:borel-cantelli}), because $\sum_{i=1}^{\infty} 2^i \cdot 2e^{-i} = \frac{4}{e-2} < \infty$, $\Alg$ correctly decides $\Pi(x)$ for all but finitely many $x \in \Dom(\Pi)$, with probability $1$ over $\UOracle$. Hence, with probability $1$ over $\UOracle$, $\Alg$ can be modified into an algorithm $\Alg'$ that agrees with $\Pi$ on every $x \in \Dom(\Pi)$, by simply hard-coding those $x$ on which $\Alg$ and $\Pi$ disagree.

Because there are only countably many $\mathsf{PromiseQMA}^{\UOracle,\COracle}$ machines, we can union bound over all $\Pi \in \mathsf{PromiseQMA}^{\UOracle,\COracle}$ to conclude that $\mathsf{PromiseQMA}^{\UOracle,\COracle} \subseteq \mathsf{PromiseBQP}^{\UOracle,\COracle}$ with probability $1$.
\end{proof}

\subsection{Pseudorandom unitaries relative to \texorpdfstring{$(\UOracle, \COracle)$}{(U, C)}}
We proceed to the second part of the oracle construction, showing that PRUs exist relative to $(\UOracle, \COracle)$. In fact, the security proof will not depend on $\COracle$: the same PRU construction is secure for \textit{any} language $\COracle$ that is independent of the randomly sampled $\UOracle$.
The PRU ensemble for a given length is supplied directly by $\UOracle$. That is, for a given length $n$, the PRU ensemble is uniform over the $2^n$ different $n$-qubit unitaries in $\UOracle_n$.

We begin with a lemma establishing that the average advantage of a polynomial-time adversary is small against our PRU construction. Here, we should think of $\{U_k\}_{k \in [N]}$ as the PRU ensemble.

\begin{lemma}
\label{lem:small_average_advantage}
Consider a quantum algorithm $\Alg^{O,U}$ that makes $T$ queries to $U = (U_1,\ldots,U_N) \in \Unitaries(D)^{N}$ and $O \in \Unitaries(D)$. For fixed $U$, define:
$$\adv(\Alg^U) \coloneqq \Pr_{k \sim [N]}\left[\Alg^{U_k,U}= 1\right] -  \Pr_{O \sim \mu_D} \left[\Alg^{O,U} = 1\right].$$
Then there exists a universal constant $c > 0$ such that:
$$\E_{U \sim \mu_D^N}\left[\adv(\Alg^U)\right] \le \frac{cT^2}{N}.$$
\end{lemma}

\begin{proof}
Our strategy is to reduce to the quantum query lower bound for unstructured search. Intuitively, if $\Alg$ could identify whether $O \in \{U_1,\ldots,U_N\}$ or not, then $\Alg$ could be modified into a quantum algorithm $\mathcal{B}$ that finds a single marked item from a list of size $N$. Then the BBBV theorem \cite{BBBV97-search} forces $T$ to be $\Omega\left(\sqrt{N}\right)$.

More formally, we construct an algorithm $\mathcal{B}^x$ that queries a string $x \in \{0,1\}^{N}$ as follows. $\mathcal{B}$ draws a unitary $V = (V_0, V_1, \ldots, V_N) \in \Unitaries(D)^{N+1}$ from $\mu_D^{N+1}$. Then, $\mathcal{B}$ runs $\Alg$, replacing queries to $O$ by queries to $V_0$, and replacing queries to $U_k \in U$ by $V_0$ if $x_k = 1$ and by $V_k$ if $x_k = 0$.

Let $e_k \in \{0,1\}^N$ be the string with $1$ in the $k$th position and $0$s everywhere else. We have that:
\begin{align*}
\E_{U \sim \mu_D^N}\left[\adv(\Alg^U)\right] &= \E_{U \sim \mu_D^N}\left[\Pr_{k \sim [N]}\left[\Alg^{U_k,U} = 1\right]\right] - \E_{U \sim \mu_D^N}\left[\Pr_{O \sim \mu_D} \left[\Alg^{O,U} = 1\right]\right]\\
&= \Pr_{k \sim [N]}\left[\mathcal{B}^{e_k} = 1\right] - \Pr\left[\mathcal{B}^{0^N} = 1\right]\\
&\le \frac{cT^2}{N}.
\end{align*}
Above, the first line applies linearity of expectation, the second line holds by definition of $\mathcal{B}$, and the third line holds for some universal $c$ by the BBBV theorem \cite{BBBV97-search}.
\end{proof}

The next lemma uses \Cref{lem:small_average_advantage} to show that the advantage of $\Alg$ is small with extremely high probability, which follows from the strong concentration properties of the Haar measure (\Cref{thm:haar_concentration}). This strengthening of \Cref{lem:small_average_advantage} will be needed to argue that the advantage remains small even after union bounding over all choices of the classical advice.

\begin{lemma}
\label{lem:small_advantage_whp}
Consider a quantum algorithm $\Alg^{O,U}$ that makes $T$ queries to $U = (U_1,\ldots,U_N) \in \Unitaries(D)^{N}$ and $O \in \Unitaries(D)$. Let $\adv(\Alg^U)$ be defined as in \Cref{lem:small_average_advantage}. Then there exists a universal constant $c > 0$ such that for any $p \ge cT^2/N$,
$$\Pr_{U \sim \mu_D^N}\left[\left|\adv(\Alg^U)\right| \ge p\right] \le 2\exp\left(-\frac{(D-2)\left(p - cT^2/N\right)^2}{96T^2}\right).$$
\end{lemma}

\begin{proof}
By \Cref{lem:query_lipschitz}, $\adv(\Alg^U)$ is $2T$-Lipschitz as a function of $U$, because $\adv(\Alg^U)$ can be expressed as the difference between the acceptance probabilities of two algorithms that each make $T$ queries to $U$. Combining \Cref{lem:small_average_advantage} and \Cref{thm:haar_concentration}, we obtain:
$$\Pr_{U \sim \mu_D^N}\left[\adv(\Alg^U) \ge p\right] \le \exp\left(-\frac{(D-2)\left(p - cT^2/N\right)^2}{96T^2}\right).$$
Similar reasoning yields the same upper bound on $\Pr_{U \sim \mu_D^N}\left[\adv(\Alg^U) \le -p\right]$, so we get the final bound (with an additional factor of 2) by a union bound.
\end{proof}

Completing the security proof of the pseudorandom unitary construction amounts to combining \Cref{lem:small_advantage_whp} with the aforementioned union bound over all possible polynomial-time adversaries.

\begin{theorem}
\label{thm:pru^u}
Let $\COracle$ be any fixed language. Then with probability $1$ over $\UOracle \sim \mathcal{D}$, there exists a family of PRUs relative to $(\UOracle, \COracle)$ with $n(\kappa) = \kappa$.
\end{theorem}

\begin{proof}
Fix an input length $n \in \Naturals$. We take the key set $\Key = \{0,1\}^n \equiv [2^n]$ and take the PRU family to be $\{U_k\}_{k \in \{0,1\}^n}$, where $\UOracle_n = (U_1,U_2,\ldots,U_{2^n}) \in \Unitaries(2^n)^{2^n}$. In words, the family consists of the $2^n$ different Haar-random $n$-qubit unitaries supplied by $\UOracle_n$.
Note that this family of unitaries has an efficient implementation relative to the oracle. This is because we can simulate an application of $U_k$ to some $n$-qubit $\ket{\psi}$ using one query to $\UOracle_n$, via $\UOracle_n\ket{k}\ket{\psi} = (I \otimes U_k)\ket{k}\ket{\psi}$. So, it remains only to show the computational indistinguishability criterion of \Cref{def:pru}.

Without loss of generality, assume the adversary is a uniform polynomial-time quantum algorithm $\Alg^{O,\UOracle,\COracle}(1^n,x)$, where $x \in \{0,1\}^{\poly(n)}$ is the advice and $O \in \Unitaries(2^n)$ is the oracle that the adversary seeks to distinguish as pseudorandom or Haar-random.

By \Cref{lem:small_advantage_whp} with $N = D = 2^n$ and $T = \poly(n)$, for any fixed $x \in \{0,1\}^{\poly(n)}$, $\Alg^{O,\UOracle,\COracle}(1^n,x)$ achieves non-negligible advantage with extremely low probability over $\UOracle$.
(The additional oracle $\COracle$ has no effect on the query complexity result because it fixed and independent of $\UOracle$.)
This is to say that for any $p = \frac{1}{\poly(n)}$:
$$\Pr_{\UOracle_n \sim \mu_{2^n}^{2^n}}\left[\left| \Pr_{k \sim [2^n]}\left[\Alg^{U_k,\UOracle,\COracle}(1^n,x) = 1 \right] - \Pr_{O \sim \mu_{2^n}} \left[\Alg^{O,\UOracle,\COracle}(1^n,x) = 1 \right] \right| \ge p\right] \le \exp\left(-\frac{2^n}{\poly(n)}\right).$$

By a union bound over all $x \in \{0,1\}^{\poly(n)}$, $\Alg^{O,\UOracle,\COracle}(1^n,x)$ achieves advantage larger than $p$ for \textit{any} $x \in \{0,1\}^{\poly(n)}$ with probability at most $2^{\poly(n)} \cdot \exp\left(-\frac{2^n}{\poly(n)}\right) \le \negl(n)$. Hence, by the Borel--Cantelli lemma (\Cref{lem:borel-cantelli}), $\Alg$ achieves negligible advantage for all but finitely many input lengths $n \in \Naturals$ with probability $1$ over $\UOracle$, as $\sum_{n=1}^\infty \negl(n) < \infty$. This is to say that $\{U_k\}_{k \in \{0,1\}^n}$ defines a PRU ensemble.
\end{proof}

We expect that using the techniques of Chung, Guo, Liu, and Qian \cite{CGLQ20-tradeoffs}, one can extend \Cref{thm:pru^u} to a security proof against adversaries with quantum advice. Some version of \cite[Theorem 5.14]{CGLQ20-tradeoffs} likely suffices. The idea is that breaking the PRU should remain hard even if $\Alg$ could query an explicit description of $O$ and explicit descriptions of $U_k$ for $k \in [2^n]$, which is a strictly more powerful model. But then this corresponds to the security game defined in \cite[Definition 5.12]{CGLQ20-tradeoffs}, except that the range of the random oracle is $\Unitaries(D)$ rather than the finite set $[M]$. Perhaps a sufficiently fine discretization of $\Unitaries(D)$ would suffice to apply the \cite{CGLQ20-tradeoffs} framework. We believe this is doable but tedious, and leave it to future work.

\subsection{Alternative oracles}
\label{sec:alternative_oracles}
An earlier version of this paper \cite{conf-version} claimed to show the same results, \Cref{thm:bqp^u=qma^u,thm:pru^u}, but relative to a different oracle. Instead of $\Oracle = (\UOracle, \COracle)$ for $\UOracle \sim \mathcal{D}$, the earlier oracle used a different classical language in place of $\COracle$; the oracle chosen was $\Oracle = (\UOracle, \mathcal{P})$ where $\mathcal{P}$ is an arbitrary $\mathsf{PSPACE}$-complete language. As noted above, PRUs still exist relative to this oracle because \Cref{thm:pru^u} works regardless of the choice of classical language. However, the claim that $\mathsf{PromiseBQP}^{\UOracle, \mathcal{P}} = \mathsf{PromiseQMA}^{\UOracle, \mathcal{P}}$ contained a bug in the proof. This incorrect step amounted to conflating the two quantities
\begin{equation}
\label{eq:wrong_quantity}
\E_{\UOracle \sim \mathcal{D}}\left[\max_{\ket{\psi}}\Pr\left[\Ver^{\UOracle, \mathcal{P}}(\ket{\psi}) = 1\right] \right]
\end{equation}
and
\[
\max_{\ket{\psi}}\E_{\UOracle \sim \mathcal{D}}\left[\Pr\left[\Ver^{\UOracle, \mathcal{P}}(\ket{\psi}) = 1\right] \right],
\]
which are not the same. Nevertheless, we conjecture that the previous proof can be restored:

\begin{conjecture}
    \label{conj:fix_bug}
    With probability $1$ over $\UOracle \sim \mathcal{D}$, $\mathsf{PromiseBQP}^{\UOracle, \mathcal{P}} = \mathsf{PromiseQMA}^{\UOracle, \mathcal{P}}$, where $\mathsf{P}$ is an arbitrary $\mathsf{PSPACE}$-complete language.
\end{conjecture}

A careful inspection of \cite{conf-version} reveals that \Cref{conj:fix_bug} could be proved by showing that the quantity in \Cref{eq:wrong_quantity} is approximable in $\mathsf{PSPACE}$. We see a possible approach to establishing this, which relies on the following well-known analogue of the polynomial method \cite{BBC+01-polynomials} for $\mathsf{QMA}$ verifiers:

\begin{proposition}[{Proved in \cite[Lemma 4]{Aar09-qma1}}]
\label{prop:qma_polynomial_method}
    Let $\mathcal{V}$ be a $\mathsf{QMA}$-verifier that receives an $m$-qubit witness and makes $T$ queries to a unitary $\Oracle$. Then there exists a matrix-valued polynomial $M(\Oracle)$ of degree $2T$ in $\Oracle$ and $\Oracle^\dagger$ such that for any $m$-qubit $\ket{\psi}$,
    \[
    \braket{\psi|M(\Oracle)|\psi} = \Pr\left[\Ver^{\Oracle}(\ket{\psi}) = 1\right].
    \]
\end{proposition}

\begin{proof}
    Without loss of generality, suppose that on input $\ket{\psi}$, $\Ver$ appends $n$ ancilla qubits initialized to $\ket{0}$, applies a unitary $U(\Oracle)$ that may involve queries to $\Oracle$, and then measures the first qubit. Then the matrix $M(\Oracle)$ is:
    \[
    (I \otimes \bra{0^n})U(\Oracle)^\dagger(\ket{1}\bra{1} \otimes I)U(\Oracle)(I \otimes \ket{0^n}),
    \]
    which clearly satisfies
    \[
    \braket{\psi|M(\Oracle)|\psi} = \Pr\left[\Ver^{\Oracle}(\ket{\psi}) = 1\right].
    \]
    Additionally, the entries of $M(\Oracle)$ are polynomials of degree $2T$ in $\Oracle$ and $\Oracle^\dagger$ because $U(\Oracle)$ is a polynomial of degree $T$ \cite{BBC+01-polynomials}, and $U(\Oracle)$ appears twice in the expression.
\end{proof}

A key observation is the following: if $p(\Oracle) \coloneqq \max_{\ket{\psi}} \Pr\left[\Ver^{\Oracle}(\ket{\psi}) = 1\right]$, then for any $k \in \Naturals$
\[
p(\Oracle)^k \le \tr(M(\Oracle)^k) \le 2^m p(\Oracle)^k.
\]
Equivalently,
\[
\frac{\tr(M(\Oracle)^k)^{1/k}}{2^{m/k}} \le p(\Oracle) \le \tr(M(\Oracle)^k)^{1/k}.
\]
So, by choosing $k$ to be sufficiently large (say, $100m$), $\tr(M(\Oracle)^k)^{1/k}$ provides an arbitrarily precise estimate of $p(\Oracle)$.
Thus, to approximate \Cref{eq:wrong_quantity}, it suffices to approximate
\[
\E_{\UOracle \sim \mathcal{D}}
\left[\tr(M(\UOracle, \mathcal{P})^k)^{1/k}\right],
\]
which we believe is achievable in $\mathsf{PSPACE}$. We first observe that, as a consequence of the concentration of the Haar measure (\Cref{thm:haar_concentration}), the above quantity should satisfy
\[
\E_{\UOracle \sim \mathcal{D}}
\left[\tr(M(\UOracle, \mathcal{P})^k)^{1/k}\right]
\approx
\E_{\UOracle \sim \mathcal{D}}
\left[\tr(M(\UOracle, \mathcal{P})^k)\right]^{1/k},
\]
as long as $\Ver$ only makes queries to $\UOracle$ in sufficiently large dimension.

Notice that $\tr(M(\UOracle, \mathcal{P})^k)$ is a polynomial of degree $2Tk$ in the entries of $\UOracle$ and $\mathcal{P}$ (and their inverses). 
Moreover, the proof of \Cref{prop:qma_polynomial_method} reveals that the coefficients of this polynomial are computable in $\mathsf{PSPACE}$, by standard path integral techniques \cite[Section 4.5.5]{NC10-book}. The main question, then, is whether one can average this polynomial over the Haar measure in $\mathsf{PSPACE}$. With some additional work, we believe this could be established via either
\begin{enumerate}[(1)]
    \item Showing that the Weingarten calculus \cite{CMN22-weingarten}, used for evaluating Haar integrals, is computable in $\mathsf{PSPACE}$, or
    \item Proving that a sufficiently strong notion of unitary $t$-design (\Cref{def:unitary_t_design}) yields an approximation of $\E_{\UOracle \sim \mathcal{D}}
\left[\tr(M(\UOracle, \mathcal{P})^k)\right]$. The challenge here is that this expression involves applications of both $\UOracle$ and $\UOracle^\dagger$, even when $\Ver$ only makes queries to $\UOracle$ in the forward direction. So, \Cref{lem:controlled_adaptive_unitary_t_design_multiplicative} does not seem applicable.
\end{enumerate}

\section*{Acknowledgments}
I thank many people for their assistance in completing this work, including: Scott Aaronson for suggestions on the writing, Amit Behera for pointing out several flaws in an earlier version of this work, Adam Bouland for numerous insightful discussions, Nick Hunter-Jones for conversations about $t$-designs, Qipeng Liu for clarifying some questions about \cite{CGLQ20-tradeoffs}, Ewin Tang for drawing my attention to \cite{HKOT23-diamond}, Shogo Yamada for identifying a bug in \Cref{thm:bqp^u=qma^u}, and Chinmay Nirkhe for discussions on rectifying said bug.

\phantomsection\addcontentsline{toc}{section}{References}
\bibliographystyle{alphaurl}
\bibliography{MainBibliography}

\newcommand{\etalchar}[1]{$^{#1}$}
\begin{thebibliography}{CGG{\etalchar{+}}23}

\bibitem[Aar05]{Aar05-postbqp}
Scott Aaronson.
\newblock Quantum computing, postselection, and probabilistic polynomial-time.
\newblock {\em Proceedings of the Royal Society A}, 461:3473--3482, 2005.
\newblock \href {https://doi.org/10.1098/rspa.2005.1546} {\path{doi:10.1098/rspa.2005.1546}}.

\bibitem[Aar09]{Aar09-qma1}
Scott Aaronson.
\newblock On perfect completeness for {QMA}.
\newblock {\em Quantum Inf. Comput.}, 9(1):81--89, jan 2009.
\newblock \href {https://doi.org/10.26421/QIC9.1-2-5} {\path{doi:10.26421/QIC9.1-2-5}}.

\bibitem[Aar16]{Aar16-barbados}
Scott Aaronson.
\newblock The complexity of quantum states and transformations: From quantum money to black holes, 2016.
\newblock \href {https://arxiv.org/abs/1607.05256} {\path{arXiv:1607.05256}}.

\bibitem[Aar18]{Aar18-shadow-tomography}
Scott Aaronson.
\newblock Shadow tomography of quantum states.
\newblock In {\em Proceedings of the 50th Annual ACM SIGACT Symposium on Theory of Computing}, STOC 2018, pages 325--338, New York, NY, USA, 2018. Association for Computing Machinery.
\newblock \href {https://doi.org/10.1145/3188745.3188802} {\path{doi:10.1145/3188745.3188802}}.

\bibitem[AE07]{AE07-designs}
Andris Ambainis and Joseph Emerson.
\newblock Quantum t-designs: T-wise independence in the quantum world.
\newblock In {\em Proceedings of the Twenty-Second Annual IEEE Conference on Computational Complexity}, CCC '07, pages 129--140, USA, 2007. IEEE Computer Society.
\newblock \href {https://doi.org/10.1109/CCC.2007.26} {\path{doi:10.1109/CCC.2007.26}}.

\bibitem[AG04]{AG04-stabilizer}
Scott Aaronson and Daniel Gottesman.
\newblock Improved simulation of stabilizer circuits.
\newblock {\em Phys. Rev. A}, 70:052328, Nov 2004.
\newblock \href {https://doi.org/10.1103/PhysRevA.70.052328} {\path{doi:10.1103/PhysRevA.70.052328}}.

\bibitem[AK07]{AK07-qcma-qma}
Scott Aaronson and Greg Kuperberg.
\newblock Quantum versus classical proofs and advice.
\newblock {\em Theory of Computing}, 3(7):129--157, 2007.
\newblock \href {https://doi.org/10.4086/toc.2007.v003a007} {\path{doi:10.4086/toc.2007.v003a007}}.

\bibitem[AKN98]{AKN98-mixed}
Dorit Aharonov, Alexei Kitaev, and Noam Nisan.
\newblock Quantum circuits with mixed states.
\newblock In {\em Proceedings of the Thirtieth Annual ACM Symposium on Theory of Computing}, STOC '98, pages 20--30, New York, NY, USA, 1998. Association for Computing Machinery.
\newblock \href {https://doi.org/10.1145/276698.276708} {\path{doi:10.1145/276698.276708}}.

\bibitem[All17]{All17-meta}
Eric Allender.
\newblock {\em The Complexity of Complexity}, pages 79--94.
\newblock Springer International Publishing, Cham, 2017.
\newblock \href {https://doi.org/10.1007/978-3-319-50062-1\_6} {\path{doi:10.1007/978-3-319-50062-1\_6}}.

\bibitem[All20]{All20-mcsp}
Eric Allender.
\newblock The new complexity landscape around circuit minimization.
\newblock In Alberto Leporati, Carlos Mart{\'\i}n-Vide, Dana Shapira, and Claudio Zandron, editors, {\em Language and Automata Theory and Applications}, pages 3--16, Cham, 2020. Springer International Publishing.
\newblock \href {https://doi.org/10.1007/978-3-030-40608-0\_1} {\path{doi:10.1007/978-3-030-40608-0\_1}}.

\bibitem[AMR20]{AMR20-random}
Gorjan Alagic, Christian Majenz, and Alexander Russell.
\newblock Efficient simulation of random states and random unitaries.
\newblock In Anne Canteaut and Yuval Ishai, editors, {\em Advances in Cryptology -- EUROCRYPT 2020}, pages 759--787, Cham, 2020. Springer International Publishing.
\newblock \href {https://doi.org/10.1007/978-3-030-45727-3\_26} {\path{doi:10.1007/978-3-030-45727-3\_26}}.

\bibitem[AQY22]{AQY22-prs}
Prabhanjan Ananth, Luowen Qian, and Henry Yuen.
\newblock Cryptography from pseudorandom quantum states.
\newblock In Yevgeniy Dodis and Thomas Shrimpton, editors, {\em Advances in Cryptology -- CRYPTO 2022}, volume 13507 of {\em Lecture Notes in Computer Science}, pages 208--236. Springer International Publishing, 2022.
\newblock \href {https://doi.org/10.1007/978-3-031-15802-5\_8} {\path{doi:10.1007/978-3-031-15802-5\_8}}.

\bibitem[BBBV97]{BBBV97-search}
Charles~H. Bennett, Ethan Bernstein, Gilles Brassard, and Umesh Vazirani.
\newblock Strengths and weaknesses of quantum computing.
\newblock {\em SIAM Journal on Computing}, 26(5):1510--1523, 1997.
\newblock \href {https://doi.org/10.1137/S0097539796300933} {\path{doi:10.1137/S0097539796300933}}.

\bibitem[BBC{\etalchar{+}}01]{BBC+01-polynomials}
Robert Beals, Harry Buhrman, Richard Cleve, Michele Mosca, and Ronald de~Wolf.
\newblock Quantum lower bounds by polynomials.
\newblock {\em J. ACM}, 48(4):778--797, Jul 2001.
\newblock \href {https://doi.org/10.1145/502090.502097} {\path{doi:10.1145/502090.502097}}.

\bibitem[BCQ23]{BCQ23-efi}
Zvika Brakerski, Ran Canetti, and Luowen Qian.
\newblock {On the Computational Hardness Needed for Quantum Cryptography}.
\newblock In Yael Tauman~Kalai, editor, {\em 14th Innovations in Theoretical Computer Science Conference (ITCS 2023)}, volume 251 of {\em Leibniz International Proceedings in Informatics (LIPIcs)}, pages 24:1--24:21, Dagstuhl, Germany, 2023. Schloss Dagstuhl -- Leibniz-Zentrum f{\"u}r Informatik.
\newblock \href {https://doi.org/10.4230/LIPIcs.ITCS.2023.24} {\path{doi:10.4230/LIPIcs.ITCS.2023.24}}.

\bibitem[BEM24]{BEM24-shrink}
Samuel Bouaziz-Ermann and Garazi Muguruza.
\newblock Quantum pseudorandomness cannot be shrunk in a black-box way.
\newblock Cryptology ePrint Archive, Paper 2024/291, 2024.
\newblock URL: \url{https://eprint.iacr.org/2024/291}.

\bibitem[Ber13]{Ber13-bayes-decision}
James~O. Berger.
\newblock {\em Statistical Decision Theory and Bayesian Analysis}.
\newblock Springer Series in Statistics. Springer New York, 2013.
\newblock \href {https://doi.org/10.1007/978-1-4757-4286-2} {\path{doi:10.1007/978-1-4757-4286-2}}.

\bibitem[BFV20]{BFV20-prs-wormhole}
Adam Bouland, Bill Fefferman, and Umesh Vazirani.
\newblock {Computational Pseudorandomness, the Wormhole Growth Paradox, and Constraints on the AdS/CFT Duality (Abstract)}.
\newblock In Thomas Vidick, editor, {\em 11th Innovations in Theoretical Computer Science Conference (ITCS 2020)}, volume 151 of {\em Leibniz International Proceedings in Informatics (LIPIcs)}, pages 63:1--63:2, Dagstuhl, Germany, 2020. Schloss Dagstuhl--Leibniz-Zentrum fuer Informatik.
\newblock \href {https://doi.org/10.4230/LIPIcs.ITCS.2020.63} {\path{doi:10.4230/LIPIcs.ITCS.2020.63}}.

\bibitem[BHH16a]{BHH16-designs-short}
Fernando G. S.~L. Brand{\~a}o, Aram~W. Harrow, and Micha\l{} Horodecki.
\newblock Efficient quantum pseudorandomness.
\newblock {\em Phys. Rev. Lett.}, 116:170502, Apr 2016.
\newblock \href {https://doi.org/10.1103/PhysRevLett.116.170502} {\path{doi:10.1103/PhysRevLett.116.170502}}.

\bibitem[BHH16b]{BHH16-designs}
Fernando G. S.~L. Brand{\~a}o, Aram~W. Harrow, and Micha{\l} Horodecki.
\newblock Local random quantum circuits are approximate polynomial-designs.
\newblock {\em Communications in Mathematical Physics}, 346(2):397--434, 2016.
\newblock \href {https://doi.org/10.1007/s00220-016-2706-8} {\path{doi:10.1007/s00220-016-2706-8}}.

\bibitem[Bor09]{Bor09-prob}
{\'E}mile Borel.
\newblock Les probabilit{\'e}s d{\'e}nombrables et leurs applications arithm{\'e}tiques.
\newblock {\em Rendiconti del Circolo Matematico di Palermo (1884-1940)}, 27(1):247--271, 1909.
\newblock \href {https://doi.org/10.1007/BF03019651} {\path{doi:10.1007/BF03019651}}.

\bibitem[BR20]{BR20-apxcount}
Aleksandrs Belovs and Ansis Rosmanis.
\newblock Tight quantum lower bound for approximate counting with quantum states, 2020.
\newblock \href {https://arxiv.org/abs/2002.06879} {\path{arXiv:2002.06879}}.

\bibitem[BS19]{BS19-binary}
Zvika Brakerski and Omri Shmueli.
\newblock ({P}seudo) random quantum states with binary phase.
\newblock In Dennis Hofheinz and Alon Rosen, editors, {\em Theory of Cryptography}, pages 229--250, Cham, 2019. Springer International Publishing.
\newblock \href {https://doi.org/10.1007/978-3-030-36030-6\_10} {\path{doi:10.1007/978-3-030-36030-6\_10}}.

\bibitem[BS20]{BS20-scalable}
Zvika Brakerski and Omri Shmueli.
\newblock Scalable pseudorandom quantum states.
\newblock In Daniele Micciancio and Thomas Ristenpart, editors, {\em Advances in Cryptology -- CRYPTO 2020}, pages 417--440, Cham, 2020. Springer International Publishing.
\newblock \href {https://doi.org/10.1007/978-3-030-56880-1\_15} {\path{doi:10.1007/978-3-030-56880-1\_15}}.

\bibitem[Can17]{Can17-prob}
F.P. Cantelli.
\newblock Sulla probabilit{\`a} come limite della frequenza.
\newblock {\em Atti Reale Academia Nazionale dei Lincei}, 26(1):39--45, 1917.

\bibitem[CGG{\etalchar{+}}23]{CGGH+23-one-way}
Bruno Cavalar, Eli Goldin, Matthew Gray, Peter Hall, Yanyi Liu, and Angelos Pelecanos.
\newblock On the computational hardness of quantum one-wayness, 2023.
\newblock \href {https://arxiv.org/abs/2312.08363} {\path{arXiv:2312.08363}}.

\bibitem[CGLQ20]{CGLQ20-tradeoffs}
Kai-Min Chung, Siyao Guo, Qipeng Liu, and Luowen Qian.
\newblock Tight quantum time-space tradeoffs for function inversion.
\newblock In {\em 2020 IEEE 61st Annual Symposium on Foundations of Computer Science (FOCS)}, pages 673--684, 2020.
\newblock \href {https://doi.org/10.1109/FOCS46700.2020.00068} {\path{doi:10.1109/FOCS46700.2020.00068}}.

\bibitem[CMN22]{CMN22-weingarten}
Beno{\^\i}t Collins, Sho Matsumoto, and Jonathan Novak.
\newblock The {W}eingarten calculus.
\newblock {\em Notices of the American Mathematical Society}, 69(5):734--745, 2022.
\newblock \href {https://doi.org/10.1090/noti2474} {\path{doi:10.1090/noti2474}}.

\bibitem[GST24]{GST20-if}
Zuzana Gavorov{\'a}, Matan Seidel, and Yonathan Touati.
\newblock Topological obstructions to quantum computation with unitary oracles.
\newblock {\em Phys. Rev. A}, 109:032625, Mar 2024.
\newblock \href {https://doi.org/10.1103/PhysRevA.109.032625} {\path{doi:10.1103/PhysRevA.109.032625}}.

\bibitem[HKOT23]{HKOT23-diamond}
Jeongwan Haah, Robin Kothari, Ryan O'Donnell, and Ewin Tang.
\newblock Query-optimal estimation of unitary channels in diamond distance.
\newblock In {\em 2023 IEEE 64th Annual Symposium on Foundations of Computer Science (FOCS)}, pages 363--390, 2023.
\newblock \href {https://doi.org/10.1109/FOCS57990.2023.00028} {\path{doi:10.1109/FOCS57990.2023.00028}}.

\bibitem[HKP20]{HKP20-classical-shadows}
Hsin-Yuan Huang, Richard Kueng, and John Preskill.
\newblock Predicting many properties of a quantum system from very few measurements.
\newblock {\em Nature Physics}, 2020.
\newblock \href {https://doi.org/10.1038/s41567-020-0932-7} {\path{doi:10.1038/s41567-020-0932-7}}.

\bibitem[HMY23]{HMY23-qpke}
Minki Hhan, Tomoyuki Morimae, and Takashi Yamakawa.
\newblock From the hardness of detecting superpositions to cryptography: Quantum public key encryption and commitments.
\newblock In Carmit Hazay and Martijn Stam, editors, {\em Advances in Cryptology -- EUROCRYPT 2023}, pages 639--667, Cham, 2023. Springer Nature Switzerland.
\newblock \href {https://doi.org/10.1007/978-3-031-30545-0\_22} {\path{doi:10.1007/978-3-031-30545-0\_22}}.

\bibitem[IR89]{IR89-permutations}
Russell Impagliazzo and Steven Rudich.
\newblock Limits on the provable consequences of one-way permutations.
\newblock In {\em Proceedings of the Twenty-First Annual ACM Symposium on Theory of Computing}, STOC '89, pages 44--61, New York, NY, USA, 1989. Association for Computing Machinery.
\newblock \href {https://doi.org/10.1145/73007.73012} {\path{doi:10.1145/73007.73012}}.

\bibitem[JLS18]{JLS18-prs}
Zhengfeng Ji, Yi-Kai Liu, and Fang Song.
\newblock Pseudorandom quantum states.
\newblock In Hovav Shacham and Alexandra Boldyreva, editors, {\em Advances in Cryptology -- CRYPTO 2018}, pages 126--152, Cham, 2018. Springer International Publishing.
\newblock \href {https://doi.org/10.1007/978-3-319-96878-0\_5} {\path{doi:10.1007/978-3-319-96878-0\_5}}.

\bibitem[KQST23]{KQST23-prs}
William Kretschmer, Luowen Qian, Makrand Sinha, and Avishay Tal.
\newblock Quantum cryptography in {A}lgorithmica.
\newblock In {\em Proceedings of the 55th Annual ACM Symposium on Theory of Computing}, STOC 2023, pages 1589--1602, New York, NY, USA, 2023. Association for Computing Machinery.
\newblock \href {https://doi.org/10.1145/3564246.3585225} {\path{doi:10.1145/3564246.3585225}}.

\bibitem[Kre21a]{conf-version}
William Kretschmer.
\newblock {Quantum Pseudorandomness and Classical Complexity}.
\newblock In Min-Hsiu Hsieh, editor, {\em 16th Conference on the Theory of Quantum Computation, Communication and Cryptography (TQC 2021)}, volume 197 of {\em Leibniz International Proceedings in Informatics (LIPIcs)}, pages 2:1--2:20, Dagstuhl, Germany, 2021. Schloss Dagstuhl -- Leibniz-Zentrum f{\"u}r Informatik.
\newblock \href {https://doi.org/10.4230/LIPIcs.TQC.2021.2} {\path{doi:10.4230/LIPIcs.TQC.2021.2}}.

\bibitem[Kre21b]{Kre21-tsirelson}
William Kretschmer.
\newblock The {Q}uantum {S}upremacy {T}sirelson {I}nequality.
\newblock {\em {Quantum}}, 5:560, October 2021.
\newblock \href {https://doi.org/10.22331/q-2021-10-07-560} {\path{doi:10.22331/q-2021-10-07-560}}.

\bibitem[KS14]{KS14-random-clifford}
Robert Koenig and John~A. Smolin.
\newblock How to efficiently select an arbitrary {C}lifford group element.
\newblock {\em Journal of Mathematical Physics}, 55(12):122202, 2014.
\newblock \href {https://doi.org/10.1063/1.4903507} {\path{doi:10.1063/1.4903507}}.

\bibitem[Kup15]{Kup15-jones}
Greg Kuperberg.
\newblock How hard is it to approximate the {J}ones polynomial?
\newblock {\em Theory of Computing}, 11(6):183--219, 2015.
\newblock \href {https://doi.org/10.4086/toc.2015.v011a006} {\path{doi:10.4086/toc.2015.v011a006}}.

\bibitem[LMW24]{LMW24-synthesis}
Alex Lombardi, Fermi Ma, and John Wright.
\newblock A one-query lower bound for unitary synthesis and breaking quantum cryptography.
\newblock In {\em Proceedings of the 56th Annual ACM Symposium on Theory of Computing}, STOC 2024, pages 979--990, New York, NY, USA, 2024. Association for Computing Machinery.
\newblock \href {https://doi.org/10.1145/3618260.3649650} {\path{doi:10.1145/3618260.3649650}}.

\bibitem[LO22]{LO22-survey}
Zhenjian Lu and Igor~C. Oliveira.
\newblock Theory and applications of probabilistic {K}olmogorov complexity.
\newblock {\em Bull. {EATCS}}, 137, 2022.
\newblock URL: \url{http://bulletin.eatcs.org/index.php/beatcs/article/view/700}.

\bibitem[Mec19]{Mec19-random-matrix}
Elizabeth~S. Meckes.
\newblock {\em The Random Matrix Theory of the Classical Compact Groups}.
\newblock Cambridge Tracts in Mathematics. Cambridge University Press, 2019.
\newblock \href {https://doi.org/10.1017/9781108303453} {\path{doi:10.1017/9781108303453}}.

\bibitem[MY22a]{MY22-owq}
Tomoyuki Morimae and Takashi Yamakawa.
\newblock One-wayness in quantum cryptography, 2022.
\newblock \href {https://arxiv.org/abs/2210.03394} {\path{arXiv:2210.03394}}.

\bibitem[MY22b]{MY22-prs}
Tomoyuki Morimae and Takashi Yamakawa.
\newblock Quantum commitments and signatures without one-way functions.
\newblock In Yevgeniy Dodis and Thomas Shrimpton, editors, {\em Advances in Cryptology -- CRYPTO 2022}, volume 13507 of {\em Lecture Notes in Computer Science}, pages 269--295. Springer International Publishing, 2022.
\newblock \href {https://doi.org/10.1007/978-3-031-15802-5\_10} {\path{doi:10.1007/978-3-031-15802-5\_10}}.

\bibitem[NC10]{NC10-book}
Michael~A. Nielsen and Isaac~L. Chuang.
\newblock {\em Quantum Computation and Quantum Information: 10th Anniversary Edition}.
\newblock Cambridge University Press, 2010.
\newblock \href {https://doi.org/10.1017/CBO9780511976667} {\path{doi:10.1017/CBO9780511976667}}.

\bibitem[Ros21]{Ros21-unitary}
Gregory Rosenthal.
\newblock Query and depth upper bounds for quantum unitaries via {G}rover search, 2021.
\newblock \href {https://arxiv.org/abs/2111.07992} {\path{arXiv:2111.07992}}.

\bibitem[Sus16a]{Sus16-addendum}
Leonard Susskind.
\newblock Addendum to computational complexity and black hole horizons.
\newblock {\em Fortschritte der Physik}, 64(1):44--48, 2016.
\newblock \href {https://doi.org/10.1002/prop.201500093} {\path{doi:10.1002/prop.201500093}}.

\bibitem[Sus16b]{Sus16-horizons}
Leonard Susskind.
\newblock Computational complexity and black hole horizons.
\newblock {\em Fortschritte der Physik}, 64(1):24--43, 2016.
\newblock \href {https://doi.org/10.1002/prop.201500092} {\path{doi:10.1002/prop.201500092}}.

\bibitem[VMS04]{VMS04-gates}
Juha~J. Vartiainen, Mikko M\"ott\"onen, and Martti~M. Salomaa.
\newblock Efficient decomposition of quantum gates.
\newblock {\em Phys. Rev. Lett.}, 92:177902, Apr 2004.
\newblock \href {https://doi.org/10.1103/PhysRevLett.92.177902} {\path{doi:10.1103/PhysRevLett.92.177902}}.

\bibitem[WBV08]{WBV08-pseudorandom}
Yaakov~S. Weinstein, Winton~G. Brown, and Lorenza Viola.
\newblock Parameters of pseudorandom quantum circuits.
\newblock {\em Phys. Rev. A}, 78:052332, Nov 2008.
\newblock \href {https://doi.org/10.1103/PhysRevA.78.052332} {\path{doi:10.1103/PhysRevA.78.052332}}.

\end{thebibliography}

\appendix
\section{PRSs with Binary Phases}
\label{app:binary_phase_prs}
In this section, we sketch a proof that a PRS construction proposed by Ji, Liu, and Song \cite{JLS18-prs} and shown secure by Brakerski and Shmueli \cite{BS19-binary} can be broken efficiently with an $\mathsf{NP}$ oracle. The PRS family is based on pseudorandom functions (PRFs). Let $\{f_k\}_{k \in \Key}$ be a PRF family of functions $f_k: \{0,1\}^n \to \{0,1\}$ keyed by $\Key$. The corresponding PRS family is the set of states $\{\ket{\varphi_k}\}_{k \in \Key}$ given by:
$$\ket{\varphi_k} \coloneqq \frac{1}{2^{n/2}} \sum_{x \in \{0,1\}^n} (-1)^{f_k(x)} \ket{x}.$$
For simplicity, suppose that each $f_k$ is \textit{balanced}, meaning that $|f^{-1}_k(0)| = |f^{-1}_k(1)| = 2^{n-1}$. Consider the quantum circuit below:
$$
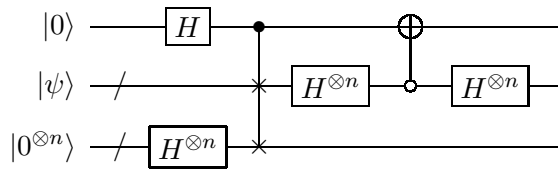
\Qcircuit @C=1em @R=.7em {
\lstick{\ket{0}} & \qw & \gate{H} & \ctrl{2} & \qw & \targ & \qw & \qw\\
\lstick{\ket{\psi}} & {/} \qw & \qw & \qswap & \gate{H^{\otimes n}} & \ctrlo{-1} & \gate{H^{\otimes n}} & \qw\\
\lstick{\ket{0^{\otimes n}}} & {/} \qw & \gate{H^{\otimes n}} & \qswap \qw & \qw & \qw & \qw & \qw
}$$
Observe that if $\ket{\psi} = \ket{\varphi_k}$, then this circuit produces the state $\ket{0}\frac{\ket{\varphi_k}\ket{+}^{\otimes n} + \ket{+}^{\otimes n}\ket{\varphi_k}}{\sqrt{2}}$ from a single copy of $\ket{\varphi_k}$. Notice that if we measure the resulting state in the computational basis, then we observe $\ket{0}\ket{x}\ket{y}$ with nonzero probability for $x, y \in \{0,1\}^n$ if and only if $f_k(x) = f_k(y)$. This is because the amplitude on this basis state is given by:
$$\bra{x}\bra{y}\frac{\ket{\varphi_k}\ket{+}^{\otimes n} + \ket{+}^{\otimes n}\ket{\varphi_k}}{\sqrt{2}} = \frac{(-1)^{f_k(x)} + (-1)^{f_k(y)}}{2^n\sqrt{2}}.$$
Furthermore, this shows that we in fact sample a uniformly random pair $(x, y)$ such that $f_k(x) = f_k(y)$.

Suppose that given a state $\ket{\psi}$ which is either pseudorandom or Haar-random, we repeat this procedure $\poly(n)$ times to obtain a list of pairs $\{(x_i, y_i)\}$. It is an $\mathsf{NP}$ problem to decide whether there exists a $k$ such that $f_k(x_i) = f_k(y_i)$ for all $i$. If $\ket{\psi} = \ket{\varphi_k}$ for some $k$ then this $\mathsf{NP}$ language always returns true, while if $\ket{\psi}$ is Haar-random, this $\mathsf{NP}$ language returns true with negligible probability, so long as we take sufficiently many samples $(x_i, y_i)$.

In the case where $f_k$ is not perfectly balanced, we simply observe that the above procedure still works with good probability so long as $f_k$ is \textit{close} to a balanced function. But PRFs must be close to balanced functions, in the sense that for most $k \in \Key$, it must be possible to change a $\negl(n)$ fraction of the outputs of $f_k$ to turn it into a balanced function. Otherwise, the PRF family could be distinguished efficiently from random functions, which are $\negl(n)$-close to balanced with high probability.

\end{document}